\documentclass[12pt]{article}
\usepackage{amssymb,amsthm,amsmath,float}
\usepackage{graphicx}
\usepackage[total={6.5in,8.8in}, top=1.0in, left=0.9in, includefoot]{geometry}
\usepackage[all,cmtip]{xy}

\newcommand{\Eq}[1]{(\ref{eq:#1})}
\newcommand{\Th}[1]{Thm.~\ref{thm:#1}}
\newcommand{\Cor}[1]{Cor.~\ref{cor:#1}}
\newcommand{\Lem}[1]{Lem.~\ref{lem:#1}}

\newcommand{\Sec}[1]{\S \ref{sec:#1}}
\newcommand{\Fig}[1]{Fig.~\ref{fig:#1}}

\newcommand{\App}[1]{Appendix~\ref{app:#1}}

\newcommand{\InsertFig}[4]{
\begin{figure}[ht]
 \centerline{
 \includegraphics[width=#4]{./figs/#1}
 }
 \caption{{\footnotesize #2}
 \label{fig:#3}}
\end{figure}}

\newcommand{\bC}{{\mathbb{ C}}}
\newcommand{\bN}{{\mathbb{ N}}}

\newcommand{\bR}{{\mathbb{ R}}}
\newcommand{\bS}{{\mathbb{ S}}}
\newcommand{\bT}{{\mathbb{ T}}}
\newcommand{\bZ}{{\mathbb{ Z}}}

\newcommand{\cU}{\mathcal{U}}

\newcommand{\eps}{\varepsilon}
\newcommand{\vphi}{\varphi}

\newcommand{\fun}{{\pi_1}}

\newcommand{\VectorField}{\mathcal{V}}
\newcommand{\Forms}{{\Lambda}}

\newcommand{\rot}{\omega}		

\newcommand{\orbit}{{\rm orb}}

\newtheorem{thm}{Theorem} 
\newtheorem{lem}[thm]{Lemma}
\newtheorem{cor}[thm]{Corollary}

\theoremstyle{definition}

\newenvironment{example}[1][]
 {
	\setlength \leftmargini {0.0em}	
	\setlength \topsep {0.7em}		
	\begin{quote}
 {\textbf{Example} #1}
}
	{\end{quote}
 }
\newcommand{\bexam}[1][:]{\begin{example}[#1]}
\newcommand{\eexam}{\end{example}}

\newcommand{\beq}[1]{\begin{equation}\label{eq:#1}}
\newcommand{\eeq}{\end{equation}}

\newcommand{\bsplit}[1]{\begin{equation}\label{eq:#1}\begin{split}}
\newcommand{\esplit}{\end{split}\end{equation}}

\title{Symmetry Reduction by Lifting for Maps}

\author{
 H.~R. Dullin$^{1}$, H.E. Lomel\'{\i}$^{2}$, and J.~D.~Meiss$^{3}$\thanks
 {
 HRD was supported in part by ARC grant DP110102001.
 HEL was supported by Asociaci\'on Mexicana de Cultura and SNI 20216.
 JDM was supported in part by NSF grant DMS-0707659.
 }
\smallskip\\
$^{1}$School of Mathematics and Statistics\\
 The University of Sydney\\
 Sydney, NSW 2006, Australia\\
 {\tt Holger.Dullin@sydney.edu.au}
\smallskip\\
$^{2}$Department of Mathematics \\
 Instituto Tecnol\'{o}gico Aut\'{o}nomo de M\'{e}xico\\
	 Mexico, DF 01000 \\
	 {\tt lomeli@itam.mx }
\smallskip\\
$^{3}$Department of Applied Mathematics\\
 University of Colorado \\
 Boulder, CO 80309-0526, USA\\
 {\tt James.Meiss@colorado.edu}
}

\date{\today}
\begin{document}
\maketitle

\begin{abstract}
\noindent
We study diffeomorphisms that have one-parameter families of continuous symmetries.
For general maps, in contrast to the symplectic case, existence of a symmetry no longer implies existence of an invariant.
Conversely, a map with an invariant need not have a symmetry.
We show that when a symmetry  flow has a global Poincar\'{e}  section there are coordinates in which the map takes a reduced, skew-product form,
and hence allows for reduction of dimensionality.
We show that the reduction of a volume-preserving map again is volume preserving.
Finally we sharpen the Noether theorem for symplectic maps.
A number of illustrative examples are discussed and the method is compared with traditional reduction techniques.
\end{abstract}

\section{Introduction}\label{sec:Introduction}

A symmetry group of a dynamical system may be discrete or continuous. The existence of discrete symmetries implies the existence of related sets of orbits and imposes constraints on bifurcations \cite{Chossat00, Golubitsky02}. Continuous symmetry often results in reduction; for example, the classical results of Sophus Lie are concerned with the reduction of order of an ODE or PDE with symmetry \cite{Olver93}. In the Hamiltonian or Lagrangian context, a continuous symmetry implies, through Noether's theorem, the existence of an invariant, and the reduction of the dynamics by two dimensions \cite{Marsden99,Marsden07,Holm09}.

In this paper, we are are interested in global reduction theory for maps that have continuous symmetries. A continuous symmetry of a map $f: M \to M$ is a vector field whose flow commutes with the map. The set of symmetries forms a Lie algebra. Symmetry reduction for maps seems to have been first studied by Maeda \cite{Maeda80, Maeda87}, who showed that a map with an $s$-dimensional, Abelian Lie algebra of symmetries can be written locally in a skew-product form
\beq{skew}
	 F(\sigma,\tau) = (k(\sigma), \tau+ \rot(\sigma)) \;,
\eeq
where $\sigma \in \bR^{n-s}$ and $\tau \in \bR^s$, in the neighborhood of any point where the rank of the symmetry group is $s$.

We will show in \Sec{mapSymmetry} that if the symmetry flow has a global Poincar\'{e}  section $\Sigma$ that is relatively closed in a manifold $M$ (that is not necessarily compact), then we can find a covering map
of the form $p:\Sigma \times \bR\to M$. If some topological conditions are satisfied, then the map $f$ has a lift $F$ to $\Sigma \times \bR$ with the skew-product form \Eq{skew}, see \Th{Reduction}. The idea relies on the fact that the topologies of $\Sigma$ and $M$ may constitute an obstruction for this lift to exist. We call this procedure \emph{reduction by lifting}; it is a global version of Maeda's result.

This global reduction procedure is distinct from the general local reduction method due to Palais \cite{Palais61}. He showed that when an $s$-dimensional symmetry group has a proper action on $M$, then there is a codimension-$s$ local neighborhood of each point, a \emph{Palais slice}, such that the group orbit of a slice forms a tube within which the orbit trivializes: there are generalized ``flow-box" coordinates. In this paper, we are primarily concerned with one-parameter symmetry groups, though our results can be extended to the case of multi-dimensional, Abelian groups. Moreover, we will show that the skew-product \Eq{skew} can be globally valid on $M$, or at least on an open, dense subset.

In the classical theory of flows with Lie symmetry groups, a flow on a manifold $M$ with symmetry group $G$ that acts properly on $M$, can be reduced to a flow on the space of group orbits, $M/G$. When the action of $G$ is free, then the group orbit space is a manifold; alternatively $M/G$ is an ``orbifold" and the reduction has singularities.

By contrast, in our reduction procedure, there is no need to assume that the action of the symmetry group is proper or free; it thus circumvents some of the problems with singular reduction. A number of examples are given in \Sec{Examples}.

One motivation of our study is to extend some results from the setting of Hamiltonian vector fields or symplectic maps to the setting of divergence-free vector fields or volume-preserving maps. Indeed, whenever a dynamical system falls into a particular structural class (symplectic, volume preserving,...), the question arises whether the reduction by appropriate symmetries can be performed so that the structure is preserved.

We will show in \Sec{VP} that when reduction by lifting is applied to volume-preserving maps, then the reduced map $k$ of \Eq{skew} is also volume preserving with respect to a natural volume form on $\Sigma$. In particular, in the three-dimensional, volume-preserving case, the reduced map $k:\Sigma\to\Sigma$ is symplectic.

As we recall in \Sec{Invariant}, Noether's theorem implies that whenever a Hamiltonian flow has a Hamiltonian symmetry there is an invariant, and conversely, that every invariant generates a Hamiltonian vector field that is a symmetry. This result also holds for symplectic twist maps with a Lagrangian generating function \cite{Logan73, Maeda81, Wendlandt97, Mansfield06}.

We will generalize a result of Bazzani \cite{Bazzani88} to show that, with the addition of a recurrence condition, Noether's theorem also applies more generally to symplectic maps. However, it is important to note that symmetries do not generally lead to invariants nor do invariants necessarily give rise to symmetries.

In \Sec{Comparison} we compare our procedure with two standard reduction procedures, using as an example a four-dimensional symplectic map with rotational symmetry.  We discuss the  advantages and shortcomings of reduction by lifting.

\section{Symmetry Reduction by Lifting}\label{sec:mapSymmetry}
In this section we investigate the conditions under which a map or flow that has a symmetry can be written in the skew-product form \Eq{skew}. Whenever this is possible,
the process reduces the effective dimension of the dynamics by one. To accomplish this reduction, we suppose that the symmetry vector field generates a complete flow that admits a global Poincar\'{e} section. We recall below the relationship between the existence of cross sections, covering spaces and fundamental groups and show how to use it to obtain the symmetry reduction.

\subsection{Invariants and Symmetries}\label{sec:notation}

We start by recalling some basic notation (see also \App{notation}) and facts about symmetries and invariants of vector fields and maps.
We suppose that $M$ is a (not necessarily compact) $n$-dimensional, path-connected manifold and denote the set of $C^1$ vector fields on $M$ by $\VectorField(M)$ and the set of $k$-forms on $M$ by $\Forms^k(M)$. The flow, $\vphi_{t}(x)$, of a vector field $X \in \VectorField(M)$ is the solution of the initial value problem
$\dot \vphi_{t}(x) = X(\vphi_{t}(x))$ with $\vphi_{0}(x)=x$. In this paper we will assume that each vector field has a complete flow; that is, $\vphi_t: M \to M $ is a diffeomorphism for all $t \in \bR$. In this case, $\vphi$ is an action of the one-dimensional Lie group $\bR$ on $M$ with composition rule $\vphi_t\circ \vphi_s=\vphi_{s+t}$.

An \emph{invariant} for a vector field $X$ is a function $I: M \to \bR$ such that
$
	L_X(I) = i_{X}dI = 0 \;,
$
where $L_X$ is the Lie derivative (see \App{notation}). Equivalently $I$ is invariant when $\vphi^*_t I = I$.
Similarly $I$ is an invariant for a map $f: M \to M$ if
\beq{MapInvariant}
	f^* I := I\circ f = I \;.
\eeq

A \emph{continuous symmetry} of a vector field $X$ is a vector field $Y$ that commutes with $X$, $[Y,X] = 0$, or equivalently, $L_Y X = 0$.
This implies that the flows of $X$ and $Y$, say $\vphi$ and $\psi$ respectively, commute \cite{Arnold78}:
\[
	\vphi_t \circ \psi_s = \psi_s \circ \vphi_t \;.
\]
Similarly, a vector field $Y$ is a symmetry of a map $f$ if
\beq{MapSymmetry}
	f^*Y = Y \;,
\eeq
where $f^*Y := ( Df(x) )^{-1} Y(f(x))$.
This means that the transformation $f$ leaves the differential equation $\dot y = Y(y)$ invariant. Since $y$ and $z=f(y)$ satisfy the same system of differential equations, the symmetry extends to the Lie group generated by the flow of $Y$:
\[
	\psi_t \circ f = f \circ \psi_t \;.
\]
Thus $f$ is an equivariant transformation of the flow.

Clearly, the collection of vector fields that are symmetries of a map form a Lie algebra. For example if $X,Y\in\VectorField(M)$ are symmetries of $f$, then $f^*[X,Y]=[f^*X,f^*Y]=[X,Y]$, so $[X,Y]$ is also a symmetry.

\subsection{Global Poincar\'{e} Sections}

We will assume $f:M \to M$ is a diffeomorphism with a symmetry vector field $Y$, \Eq{MapSymmetry}. Standing assumptions are that $M$ is an $n$-dimensional, path-connected manifold, and that $Y$ has a complete flow with a global Poincar\'{e}  section.

Recall that a relatively closed, codimension-one submanifold $\Sigma \hookrightarrow M$
is a \emph{Poincar\'{e} section} or cross section of a flow $\psi$ if it is transverse to the flow, and is a \emph{global Poincar\'{e} section} of a complete flow if every orbit of the flow has both forward and backward transversal intersections with $\Sigma$, see the left pane of \Fig{cross}. Recall that $\Sigma$ is relatively closed in $M$ if, given a sequence in $\Sigma$ that converges in $M$, the sequence also
converges in $\Sigma$.

\InsertFig{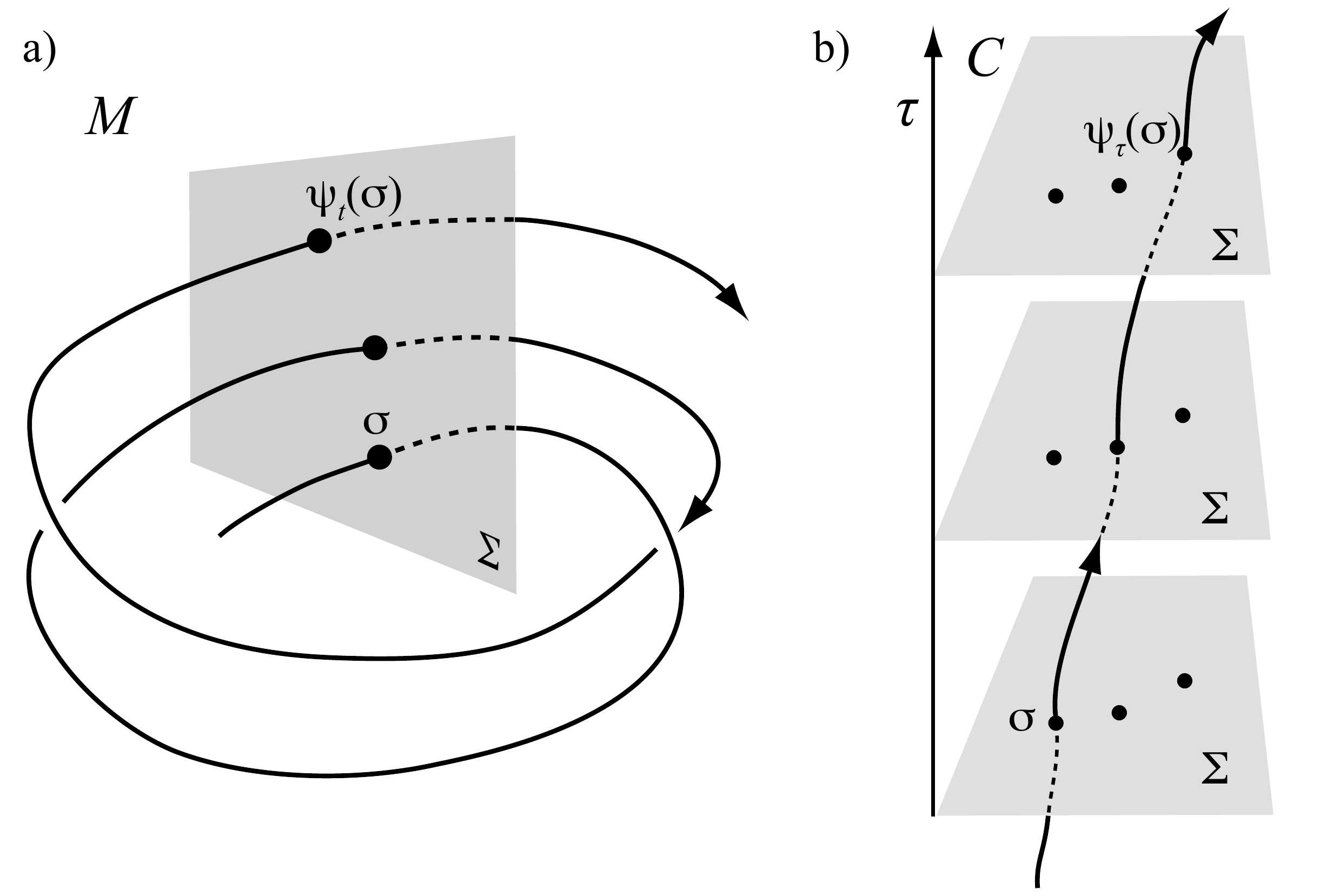}{A global Poincar\'{e} section $\Sigma$ has two representations that are equivalent by \Th{covering}:
	\emph{a}) as an embedding in the original manifold $M$, and
	\emph{b}) lifted to a cover $C = \Sigma \times \bR$.
In both cases, an orbit of the flow $\psi$ is drawn.}{cross}{3.5in}

When $\Sigma$ is a global Poincar\'{e} section for $\psi$, the first return time to $\Sigma$, $T: \Sigma \to \bR^+$, is the smallest positive number such that for, each $\sigma \in \Sigma$, $\psi_{T(\sigma)}(\sigma)\in\Sigma$. In fact, $T$ is continuous.
The first return map $r_\psi: \Sigma \to \Sigma$, also known as the Poincar\'{e} map, is defined by
\beq{PoincareMap}
	r_\psi(\sigma) := \psi_{T(\sigma)}(\sigma) \;;
\eeq
it is a diffeomorphism. The sequence of return times $T_n$, $n \in \bZ$, is defined so that $T_0(\sigma)=0$,
\[
	T_{n+1}(\sigma):=T_n(\sigma)+T(r_\psi^n(\sigma)) \;, \quad n \ge 0 \;,
\]
and $T_{-n}(\sigma)=-T_n(r_\psi^{-n}(\sigma))$.

Using this notation, the iterates of the first return map are
\beq{PoincareIterates}
	r_\psi^n(\sigma)=\psi_{T_n(\sigma)}(\sigma) \;,
\eeq
for all $n \in \bZ$. Note that the sequence of return times is strictly increasing with $n$, and must be unbounded when the section $\Sigma$ is relatively closed in $M$.
Indeed, if for some $\sigma_0$ the sequence $T_n(\sigma_0)$ were bounded then, since it is increasing, it would converge. Consequently, the difference $T_{n+1}(\sigma_0)-T_n(\sigma_0) \to 0$ and therefore by definition
\[
	T(r_\psi^n(\sigma_0))\to 0
\]
as $n\to\infty$. Thus $r_\psi^n(\sigma_0)=\psi_{T_n(\sigma_0)}(\sigma_0)\to\sigma^*$ would also converge, and since $\Sigma$ is relatively closed, we could conclude that $\sigma^*\in\Sigma$. By continuity, $T(\sigma^*)=0$, but this contradicts the fact that the return times are always positive. A similar argument shows that $T_n(\sigma)$ is unbounded when $n\to -\infty$.

A global Poincar\'{e} section $\Sigma$ can be viewed both as a submanifold of $M$, and---as illustrated in the right pane of \Fig{cross}---as the base of a covering space for $M$.
Recall that a manifold $C$ is a cover of $M$ if there is a differentiable, surjective function $p:C\to M$ such that each $m\in M$ has a neighborhood $\cU$ for which $p^{-1}(\cU)$ is the disjoint union of a countable number of open sets in $C$, each of which is diffeomorphic, via $p$, to $\cU$. The covering map $p$ is necessarily a local diffeomorphism. A covering space is a fiber bundle with a discrete fiber.

We will often think of $\Sigma$ both as a submanifold of $M$ and a space in its own right. Technically, we define an inclusion map $\iota:\Sigma\hookrightarrow M$ to express the embedding, but when there is little risk of confusion, we will let $\Sigma \subset M$ denote $\iota(\Sigma)$.

A fundamental tool for the rest of the paper is the following theorem of Schwartzman that shows how global Poincar\'{e} sections are related to covering spaces.

\begin{thm}[\cite{Schwartzman62}]\label{thm:covering}
A relatively closed submanifold $\Sigma$ is a global Poincar\'{e} section of a flow $\psi$ on a manifold $M$
if and only if the map $p:\Sigma\times\bR\to M$ defined by
\beq{covering}
	p(\sigma,\tau)=\psi_\tau(\sigma)
\eeq
is a smooth cover of $M$ with an infinite, cyclic group of deck transformations.
\end{thm}

Recall that a deck transformation of $p$ is a map $\Delta:C\to C$ on the covering space such that $p \circ \Delta = p$; that is, deck transformations are the lifts of the identity \cite{Dieck08}.
For the global Poincar\'{e} section $\Sigma$, is easy to verify that
\begin{equation}\label{eq:Delta}
	\Delta(\sigma,\tau)=(r_\psi(\sigma), \tau-T(\sigma)) \;,
\end{equation}
is a deck transformation on $\Sigma \times \bR$. Taking into account \Eq{PoincareIterates}, the iterates of $\Delta$ satisfy
\begin{equation}\label{eq:deck}
	\Delta^n(\sigma,\tau)=(r_\psi^n(\sigma), \tau-T_n(\sigma))
\end{equation}
and are also deck transformations. Hence, the collection
$\{ \Delta^n: n\in \bZ \}$, is the cyclic group mentioned in \Th{covering}.
We next note that the orbits of $\Delta$ correspond to the fibers of the cover.
\begin{lem}\label{lem:Fiber}
For each $x \in M$, the fiber, $p^{-1}(x)$, of the covering map $p$ is an orbit of $\Delta$. In other words, if $p(\sigma_0,\tau_0) = x$, then
\[
	p^{-1}(x)=\orbit_\Delta(\sigma_0,\tau_0)
		:= \{\Delta^n(\sigma_0,\tau_0): n\in\bZ\} \;.
\]
\end{lem}

\begin{proof}
For any point $(\sigma',\tau') \in \orbit_\Delta(\sigma_0,\tau_0)$, there is an $n$ such that $(\sigma',\tau') = \Delta^n(\sigma_0,\tau_0)$, and since $\Delta^n$ is a deck transformation, $p(\sigma',\tau') = p \circ \Delta^n(\sigma_0,\tau_0) = p(\sigma_0,\tau_0) = x$.
Conversely, suppose that $(\sigma^\prime,\tau^\prime)\in p^{-1}(x)$, i.e.,
\(
	\psi_{\tau^\prime}(\sigma^\prime)=x=\psi_{\tau_0}(\sigma_0),
\)
so that $\sigma'= \psi_{\tau_0-\tau'}(\sigma_0) \in \Sigma$.
Assuming, without loss of generality, that $\tau_0>\tau^\prime$, then since the sequence $T_n(\sigma_0)$ is strictly increasing and unbounded, there must be an $n\in \bN$ so that
\[
T_n(\sigma_0)\leq \tau_0-\tau'<T_{n+1}(\sigma_0) \;.
\]
If we let $\sigma_n=r_\psi^n(\sigma_0)$ and $\kappa=\tau_0-\tau' -T_n(\sigma_0)$,
then $0\leq \kappa<T(\sigma_n)$ and $\sigma'=\psi_\kappa(\sigma_n)$. By definition of
the first return time, the only possibility is that $\kappa=0$.

Therefore, $\tau_0-\tau' = T_n(\sigma_0)$. Consequently, $\sigma' = r^n_\psi(\sigma_0)$
and thus $(\sigma',\tau') = \Delta^n(\sigma_0,\tau_0)$.
\end{proof}

\subsection{Symmetry Reduction by Lifting}
We now ask whether it is possible to lift $f$ to a cover $C$ in which the map takes the skew-product form \Eq{skew}. Recall that a lift of $f$ is a map
$F:C\to C$ such that
\beq{liftDef}
	f\circ p=p\circ F \;.
\eeq
In particular, any deck transformation is the lift of the identity on the base.
Lifts are not unique, but all lifts are the same up to deck-transformations:
if $G$ is another lift of $f$ then $ G  = \Delta^m F$ for some integer $m$.

Suppose that $x_0,y_0\in C$ are two points in the cover such that $f(p(x_0))=p(y_0)$.
A necessary and sufficient condition \cite{Massey} for the existence of a lift is that
\begin{equation}\label{eq:lift}
	f_*p_*\fun\left(C,x_0\right)\subseteq p_*\fun\left(C,y_0\right) \;,
\end{equation}
where $\fun(C,x_0)$ is the fundamental group of $C$ based at $x_0$. 
In addition, if the lift exists, it can be taken to satisfy $F(x_0)=y_0$.\footnote
{Recall that the fundamental group of a topological space $M$ depends on the choice of a base point $\xi_0$. However, the choice of base point is irrelevant provided the space $M$ is path-connected.}
In other words, each loop in $M$ that is a projection of a noncontractable loop in $C$ must map, under $f$, to another loop in the same projection. This is a nontrivial requirement because $p_*(\fun(C,x_0))$ is necessarily a subgroup of $\fun(M,x_0)$. Note, however that the requirement would be trivially satisfied if the
cover were simply connected, since in that case $\fun(C,x_0)$ would be trivial.

\InsertFig{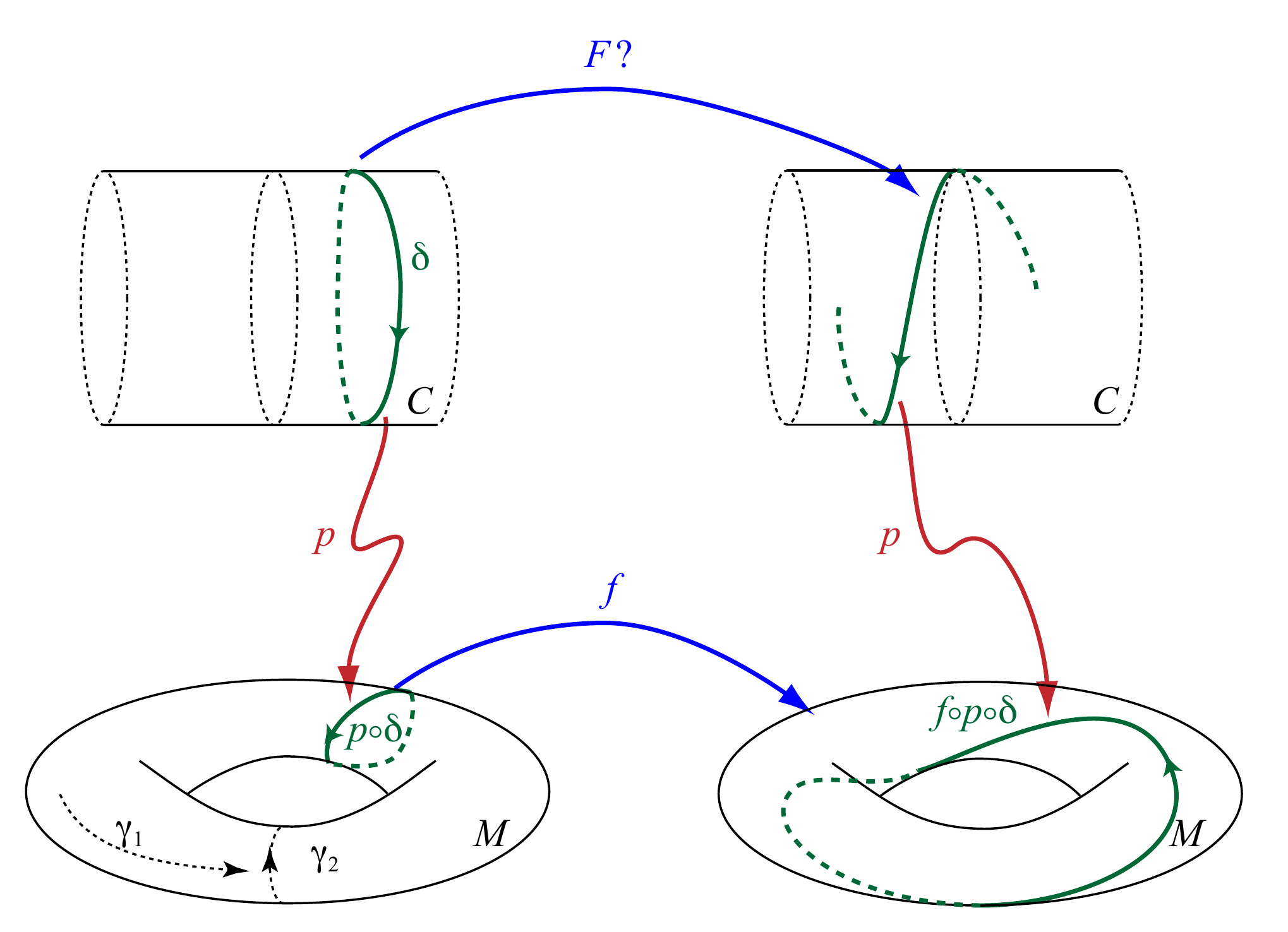}{Failure of \Eq{lift} in an attempt to construct a lift of the cat map to the cylinder.}{lift}{3.8in}

The necessity of the condition \Eq{lift} is easy to illustrate. Consider the example of Arnold's cat map on the two-torus $M=\bR^2/\bZ^2$: $f(x,y)=(2x+y,x+y)$. Does $f$ have a lift to the cylinder $C=\bR\times\bS^1$ with the natural covering map, $p(z,\tau) = (z \mod 1, \tau)$?
The fundamental group $\fun(M,0)$ has two generators, say $\gamma_1,\gamma_2$; however, $\fun(C,0)$ has only one generator, say $\delta$. Since $p_*(\delta)=\gamma_2$, and $f_*(\gamma_2)=\gamma_1+\gamma_2$, the condition
\Eq{lift} is not satisfied. The geometric reason this fails is illustrated in \Fig{lift}.
Any loop $f\circ p\circ\delta$ has the homotopy type of $\gamma_1+\gamma_2$. However,
any lift of $f\circ p\circ\delta$ can not be a loop on $C$.

We will now give conditions for existence of a lift of $f$ on the cover $C=\Sigma\times\bR$ with $p:C\to M$ as in \Eq{covering}. Let $\Psi$ denote the trivial flow on $\Sigma\times\bR$,
\beq{trivial}
	\Psi_t(\sigma,\tau)=(\sigma,\tau+t) \;.
\eeq
Clearly, $\Psi$ commutes with $\Delta$ and
\(
	p\circ \Psi_t(\sigma,\tau)=p(\sigma,\tau+t)
		=\psi_{\tau+t}(\sigma)=\psi_t\circ p(\sigma,\tau)
\)
so that $\Psi$ is a lift of $\psi$, in the sense that each $\Psi_t$ is a lift of $\psi_t$, for all $t$.

\begin{thm}[Symmetry Reduction]\label{thm:Reduction}
Suppose that a diffeomorphism $f:M\to M$ has a symmetry $\psi_t$ with global Poincar\'{e} section $\Sigma$
that is relatively closed in $M$ and let $p:\Sigma\times \bR\to M$ be the induced cover \Eq{covering}.
Let $\sigma_0\in\Sigma$ and $(\bar{\sigma},\bar{\tau})\in p^{-1}(f(\sigma_0))$.

Then the following statements are equivalent.
\begin{enumerate}
  \item\label{part1} %
  There exists a lift $F:\Sigma \times \bR\to \Sigma \times \bR$ of $f$ to the cover of the form
\beq{skewProduct}
 F: \left\{ \begin{array}{l}
			\sigma' = k(\sigma) \;, \\
			\tau' = \tau + \rot(\sigma) \;,
		 \end{array} \right.
\eeq
where $k$ is a diffeomorphism of $\Sigma$, and $k(\sigma_0)=\bar{\sigma}$, $\rot(\sigma_0)=\bar{\tau}$.
\item\label{part2} If  $\iota:\Sigma\hookrightarrow M$ is the standard inclusion, then 
 \beq{groupSubset}
	f_*\iota_*\fun\left(\Sigma,\sigma_0\right)\subseteq
            p_*\fun\left(\Sigma\times\bR,(\bar{\sigma},\bar{\tau})\right) \;.
\eeq
\end{enumerate}
\end{thm}

\begin{proof}
\ref{part1})$\Longrightarrow$\ref{part2})

By \Eq{skewProduct}, $F(\sigma,0)=(k(\sigma),\rot(\sigma))$, and since $F$ is a lift of $f$, $f(\sigma)=f\circ p(\sigma,0)=p(k(\sigma),\rot(\sigma))$.
Let $G:\Sigma\to\Sigma\times\bR$ be given by $G(\sigma):=(k(\sigma),\rot(\sigma))$.
Then $p\circ G=f\circ \iota$ and $G(\sigma_0)=(\bar{\sigma},\bar{\tau})$. Consequently $f_*\iota_*\fun\left(\Sigma,\sigma_0\right)=p_*G_*\fun\left(\Sigma,\sigma_0\right)$, and since $G_* \fun(\Sigma,\sigma_0) \subseteq \fun\left(\Sigma\times\bR,(\bar{\sigma},\bar{\tau})\right)$, and $p_*$ is injective, this directly implies \ref{part2}).

\ref{part2})$\Longrightarrow$\ref{part1})
Under the condition \Eq{groupSubset}, standard theorems of algebraic topology \cite{Massey} imply that there exists a function $G:\Sigma\to\Sigma\times\bR$ such that $p\circ G=f\circ \iota$. Define $k$ and $\omega$ by $G(\sigma)=(k(\sigma),\rot(\sigma))$. Since $f\circ \iota(\sigma_0)=p(\bar{\sigma},\bar{\tau})=f(\sigma_0)$, we can take $G(\sigma_0) = (k(\sigma_0), \omega(\sigma_0)) = (\bar{\sigma},\bar{\tau})$.

Now, for each $\tau\in\bR$, we define $F(\sigma,\tau):=\Psi_\tau(G(\sigma))$, where $\Psi$ is the trivial
flow \Eq{trivial}. We claim that $F$ is a lift of $f$, i.e., it satisfies \Eq{liftDef}. To show this, we notice that
\[
p\circ F(\sigma,\tau)=p\circ\Psi_\tau(G(\sigma))=\psi_\tau\circ p(G(\sigma))=
\psi_\tau(f(\sigma))=f(\psi_\tau(\sigma))=f\circ p(\sigma,\tau).
\]
\end{proof}

\InsertFig{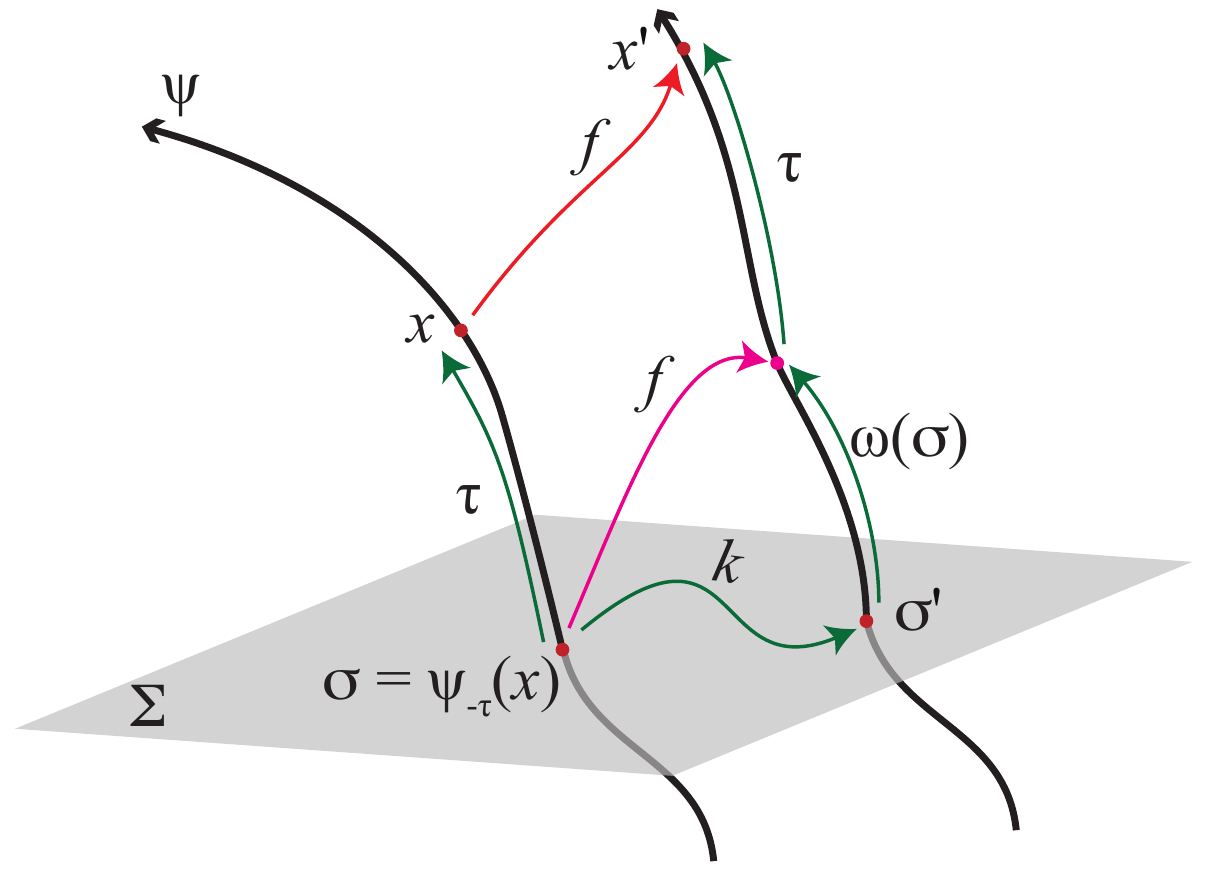}{Showing the ideas behind the skew-product form \Eq{skewProduct}.}{SkewProduct}{3.5in}

The geometry underlying the reduction \Eq{skewProduct} is illustrated in \Fig{SkewProduct}. Theorem \ref{thm:Reduction} can be seen to be equivalent to a commuting diagram.
Let $j:\Sigma\to \Sigma \times \bR$ be the inclusion given by $j(\sigma):=(\sigma,0)$.
When $\iota:\Sigma\hookrightarrow M$ is the standard inclusion,
then  $p\circ j=\iota$. 
Let $\Pi:\Sigma\times \bR\to \Sigma$ be the canonical projection.
Clearly, $\Pi\circ j=id_\Sigma$.
With this notation, \Th{Reduction} can be represented by the commutative diagram
\bsplit{commutative}
 \xymatrixcolsep{4pc}
 \xymatrix{
	&\Sigma\ar[r]^k&\Sigma\\
	\Sigma \ar[r]^j\ar@{_{(}->}[dr]_\iota\ar[ur]^{id_\Sigma}&
	\Sigma\times\bR \ar[u]_\Pi\ar	[r]^F\ar[d]_p&
	\Sigma\times\bR \ar[u]_\Pi\ar[d]_p\\
	&M\ar[r]^f&M}
\end{split}\eeq
In particular, the reduced map $k$ on $\Sigma$ can be written as
\beq{reducedMap}
	k=\Pi\circ F\circ j \;.
\eeq
Furthermore, the choice of  $\sigma_0$ and $(\bar{\sigma},\bar{\tau})$
such that $f(\sigma_0)=\psi_{\bar{\tau}}(\bar{\sigma})$, fixes the choice of
$k$ and $\rot$.

A simplification in  \Th{Reduction} occurs if $f$ has a fixed point at a point $\sigma_0 \in \Sigma$. Since the inclusion $j$ induces an isomorphism of the fundamental groups,
$j_*:\fun\left(\Sigma,\sigma_0\right)\to\fun\left(\Sigma\times\bR,(\sigma_0,0)\right)$,
\[
	\iota_*\fun\left(\Sigma,\sigma_0\right)
		=p_*j_*\fun\left(\Sigma,\sigma_0\right)
		=p_*\left(\fun\left(\Sigma\times\bR,(\sigma_0,0)\right)\right) \;.
\] 
Upon taking $\bar{\tau}=0$ and $\bar{\sigma}=\sigma_0$, the condition \Eq{groupSubset} then reduces to
\beq{newSubset}
	f_*\iota_*\fun\left(\Sigma,\sigma_0\right)	\subseteq
	\iota_*\fun\left(\Sigma,\sigma_0\right)\;.
\eeq

Finally, we note a convenient equation determining $k$ and $\omega$ is obtained by combining the definition of the lift \Eq{liftDef}, the form of the covering map \Eq{covering}, and the skew-product form \Eq{skewProduct} to obtain
$
	f( \psi_\tau( \sigma) ) = \psi_{\tau + \rot(\sigma)} (  k(\sigma) ) \;.
$
Since $f$ commutes with $\psi_\tau$ we find
\beq{kDefine}
	k(\sigma) = \psi_{-\rot(\sigma)} (  f(\sigma)  ) \in \Sigma \;.
\eeq
This determines $k$ and $\rot$ (up to the choice of lift) by the requirement that the right side must be in $\Sigma$.


\subsection{Deck symmetries and a homotopy invariant}
In this section we will show that lifts of maps with symmetries have a homotopy invariant.

\begin{lem}\label{lem:homotopy}
 Suppose that $f$ and its lift $F$ satisfy the hypotheses of \Th{Reduction}, and that $\Delta$ is the deck transformation defined in \Eq{Delta}. Then there exists $m\in \bZ$ such that
$F\circ \Delta=\Delta^m\circ F$ and
 the following identities are satisfied
\begin{align*}
k\circ r_\psi&=r_\psi^m\circ k \;,\\
\rot(\sigma)-\rot(r_\psi(\sigma))&=T_m(k(\sigma))-T(\sigma) \;.
\end{align*}
The integer $m$ is a homotopy invariant of the map $f$.
In addition, if $F$ is homotopic to the identity then $m=1$.
\end{lem}

\begin{proof}
Consider the map $G=F\circ \Delta\circ F^{-1}$. Clearly $p\circ G=p$, so $G$ has to be a deck transformation. This implies that there exists $m\in \bZ$ such that $G=\Delta^m$.

The integer $m$ is independent of the choice of lift. Indeed, suppose that $\tilde F$ is another lift of $f$, then since $p \circ F = p \circ \tilde F$, it follows from \Lem{Fiber} that there is an integer $j$ such that for any $(\sigma,\tau)$, $\tilde F(\sigma,\tau) = \Delta^j \circ F(\sigma,\tau)$. By continuity, $j$ is independent of $(\sigma,\tau)$, and consequently any two lifts differ at most by a deck transformation. Thus $\tilde F \circ \Delta = \Delta^{j+m} F = \Delta^m \tilde F$.

If $F$ is homotopic to the identity then $\Delta^{-1}\circ G$ is a deck transformation that is homotopic to the identity. The only possibility is that $\Delta^{-1}\circ G$ is the identity and therefore $m=1$. Similarly since the integer $m$ is independent of the lift, continuity implies any two homotopic maps will have the same value of $m$.

For the rest, it is enough to use \Eq{deck} in the skew product \Eq{skewProduct}.
\end{proof}

\subsection{Reduction of flows}\label{sec:FlowReduction}

Theorem \ref{thm:Reduction} also applies to a flow $\varphi$, with a symmetry $\psi$.
As in the theorem, we assume that $\psi$ has a global Poincar\'{e} section $\Sigma$.

\begin{cor}[Symmetry Reduction of Flows]\label{cor:FlowReduction}
Suppose that $\psi,\vphi$ are a pair of commuting flows on $M$ and that $\psi$ has a global Poincar\'{e} section $\Sigma$ that is relatively closed in $M$. Then there is a lift $\Phi_t$ of $\varphi_t$ to $\Sigma \times \bR$ of the form
\beq{FlowReduction}
	\Phi_t(\sigma,\tau)=(k_t(\sigma), \tau+\rot(\sigma,t)) \;.
\eeq
In addition, $\Phi_t$ can be chosen to be a flow on $\Sigma\times\bR$ and hence the functions $k_t$ and $\rot(\cdot,t)$ satisfy
\begin{align*}
    k_{t+s}&=k_t\circ k_s,\\
    \rot(\sigma,t+s)&=\rot(\sigma,t)+\rot(k_t(\sigma),s) \;.
\end{align*}
\end{cor}

\begin{proof}


Just as in the definition \Eq{covering} of the cover $p$, let $q:\Sigma\times\bR\to M$ denote the map $q(\sigma,t)=\varphi_t(\sigma)$. We notice that
$q\circ j=p\circ j=\iota$. Hence for each $\sigma_0\in\Sigma$,
\[
	q_*\fun\left(\Sigma\times\bR,(\sigma_0,0)\right)=
		q_*j_*\fun\left(\Sigma,\sigma_0\right)=
		p_*j_*\fun\left(\Sigma,\sigma_0\right)=
		p_*\fun\left(\Sigma\times\bR,(\sigma_0,0)\right) \;.
\]
Consequently, there exists a map $Q:\Sigma\times\bR\to\Sigma\times\bR$ such that
$q= p \circ Q$. Equivalently, there exist functions $k(\sigma,t)$ and $\rot(\sigma,t)$ such that $Q(\sigma,t)=(k(\sigma,t),\rot(\sigma,t))$ and
$\varphi_t(\sigma)=\psi_{\rot(\sigma,t)}(k(\sigma,t))$. Denoting $k_t:=k(\cdot,t)$, define $\Phi_t$ by \Eq{FlowReduction}.
For each $t$, $\Phi_t$ is clearly a lift of $\varphi_t$, since
\[
	p\left(\Phi_t(\sigma,\tau)\right)=\psi_{\tau+\rot(\sigma,t)}(k(\sigma,t))=
		\psi_{\tau}\varphi_t(\sigma)=
		\varphi_t\psi_{\tau}(\sigma)=
		\varphi_t\left(p(\sigma,\tau)\right).
\]

We must show that it is possible to choose $\Phi_t$ so that $\Phi_0 = id_{\Sigma \times \bR}$,
and $\Phi_{t+s} = \Phi_t \circ \Phi_s$. Note first that $\Phi_0$ is a lift of the identity map
and so it must be a deck transformation. If $\Phi_0$ is not itself the identity, then we can
replace $\Phi_t$ by $\Phi_t \circ \Phi_0^{-1}$. By \Lem{homotopy}, since $\Phi_t$ is  a map
that is homotopic to the identity, it commutes with deck transformations; thus the new $\Phi_t$
is still a lift of $\vphi_t$, and satisfies $\Phi_0 = id_{\Sigma \times \bR}$.

Now, for any $s \in \bR$, let $G_t = \Phi_{t+s} \circ \Phi_t^{-1} \circ \Phi_s^{-1}$.
Note that $p \circ G_t = p$, so that $G_t$ is a deck transformation. Moreover, since
$G_0 = id_{\Sigma \times \bR}$, $G_t$ is homotopic to the identity, and thus must be the
identity. In this way we conclude that $\Phi_t$ satisfies the group property, and is a flow.
\end{proof}

The implication of \Cor{FlowReduction} is that $k_t$ is a flow on the section $\Sigma$.

\subsection{Symmetry with an Invariant}\label{sec:InvariantSymmetry}

Recall that generally the existence of a symmetry does not imply that of an associated invariant. However, it is possible that both the map and the symmetry do share an invariant. In that case, we can see that the reduced map has the same invariant.
\begin{lem}\label{lem:Invariant}
 Suppose $f:M\to M$ is a map with a symmetry $\psi$ and  the hypotheses
 of \Th{Reduction} hold. If $I:M\to \bR$ is an invariant of $f$ that is also an invariant of $\psi$, then $\iota^*I$ is an invariant of $k$.
\end{lem}
\begin{proof}
 Notice that
 \(
 F^*p^*I=p^*f^*I=p^*I.
 \)
 Moreover, for any $(\sigma,\tau)\in\Sigma\times\bR$, \Eq{covering} implies
 \[
 	p^*I(\sigma,\tau)=I(\psi_\tau(\sigma))=I(\iota(\sigma))=
 		I(\iota(\Pi(\sigma,\tau)))=\Pi^*\iota^*I(\sigma,\tau) \;.
 \]
Using the expression \Eq{reducedMap} we can then conclude that
 \[
	k^*\iota^*I = j^*F^*\Pi^*\iota^*I
		        = j^*F^*p^*I = j^*p^*I
		        = \iota^*I \;.
 \]
\end{proof}
Examples that have both invariants and symmetries will be given in \Sec{Examples}.

\subsection{Circle actions}\label{sec:CircleAction}

We have assumed that the symmetry vector field $Y$ generates a flow $\psi_t$ that  corresponds to an action of the Lie group $\bR$ on the manifold $M$. However, in many cases, this action is periodic, and can thus be thought of as an action of the group $\bS^1$ on $M$. Typically, the temporal period of the orbit of $\psi_t$ will depend upon the point. Nevertheless, if $\psi_t$ has a global section $\Sigma$, and every orbit is periodic, then some iterate of the Poincar\'{e} map, $r_\psi$, \Eq{PoincareMap}, must be the identity, because the flow returns to the original point on the same circle. In this case there is a smallest positive integer $\ell$ such that $r_\psi^\ell= id_{\Sigma}$, for all $\sigma \in \Sigma$.

The existence of a global section for a circle action $\psi_t$ on $M$ is a
strong restriction on the topology of the manifold. Indeed, since the $\ell$-th iterate of the Poincar\'e map is the identity, a $\bZ_\ell \equiv \bZ \setminus \ell \bZ$ covering of $M$ is a trivial bundle over $\bS^1$, i.e.,\ $M = ( \Sigma \times \bS^1) / \bZ_\ell$.

In general the existence of a (free and proper) circle action gives $M$ the structure of a principal $\bS^1$-fiber bundle, though it need not be trivial. A classical example is provided by the Hopf fibration of $\bS^3$. At first it seems as if the  reduction  of \Th{Reduction} could not be done for a symmetry that is a non-trivial bundle that is not a discrete quotient of a direct product with $\bS^1$. However, one can often modify $M$ in a way that it acquires the necessary topology. This may be achieved by removing parts of $M$ that are invariant under the symmetry flow and under the map; examples are given \Sec{Examples} and \Sec{Comparison}. Some global topology may be lost because the modified $M$ will in general not be compact.

The question can also be inverted: how can one possibly achieve reduction {\em and}
reconstruction in the case where the symmetry group action induces a non-trivial fiber
bundle? The semi-direct product structure of \Eq{skewProduct} requires that the dynamics in the fiber are driven by the dynamics in the base. However, if it is not possible to globally define an origin in the $\bS^1$ fiber, this cannot be done.

\subsection{Orbit spaces and global Poincar\'{e} sections}\label{sec:OrbitSpaces}

Classical symmetry reduction begins with a smooth manifold $M$ (say without boundary, or even better a compact manifold) and a Lie group $G$ with a smooth and proper $G$-action $\Phi : G \times M \to M$. Two points on the same group orbit are now considered equivalent, and symmetry reduction means to study dynamics on the \emph{orbit space} $M/G$. In general this quotient is not a manifold, but just an orbifold. Each point $x \in M$ has an associated isotropy group, namely the set of elements of $G$ that fixes $x$,
\[
	G_x=\{g\in G: g(x)=x\} \;.
\]

The conjugacy class of $G_x$ is called the orbit type and gives a stratification of $M$ into types. This induces a stratification of the orbit space $M/G$ with a smooth structure \cite{Pflaum01}. When the group action is free (or slightly more generally, if the isotropy group is the same for each point) the orbit space is a manifold, and $M$ has the structure of a principal bundle over the orbit space. This is the standard setting for symmetry reduction. Reconstruction, i.e., the study of the dynamics in the fiber, may still be challenging globally because the fiber bundle may not be trivial: it may not have a global section.

Many interesting group actions are not free. When the Lie group is compact, singular reduction is still possible, where ``singular'' reminds us that the reduced space in this case is not a manifold. The standard tool is the Hilbert basis of invariants of the group action, which together with all relations and inequalities give an accurate description of the singular reduced space, see e.g.~\cite{Chossat02,Pflaum01}. We will use this approach in passing in some of the examples in \Sec{Examples} and \Sec{Comparison}.

The principal orbit bundle (or principal stratum) is the stratum that is open and dense in $M$ \cite{Pflaum01,Chossat00,Golubitsky88}. Moreover, its quotient by $G$ is connected \cite[Thm 4.3.2]{Pflaum01}. In many cases, the isotropy group of the principal orbits is the trivial group. Within the principal stratum there is a subgroup of $G$ that acts freely, so we return to the simpler case of a free action, albeit on the smaller space obtained by removing certain closed subsets from the original $M$. The fact that the symmetry commutes with the map $f$ implies that every stratum of $M$ is a forward invariant set of $f$. When $f$ is invertible, this set is both forward and backward invariant, but for non-invertible maps there may be points outside
the invariant set that map into it.

Hence when $f$ is a diffeomorphism, the principal stratum is an open subset of full dimension that is invariant under $f$. In the examples to follow, the ``global Poincar\'e section" of the symmetry is often taken to be a section on this principal stratum.

For example, consider the circle action on $\bC^2$ given by
\beq{circleAction}
	G = \{\psi_\tau( z_1, z_2) = ( e^{il \tau} z_1, e^{im \tau } z_2)|\; (z_1,z_2) \in \bC^2,\; \tau \in \bS^1\}
\eeq
for positive integers $l$ and $m$. Note that if $\psi_t$ is a symmetry of a map $f$, then
$f$ must have, e.g., the invariant $\Im( \bar z_1^m z_2^l)$, where $ \Im(z)$ represents the imaginary part of $z$.

When $l=m=1$ the only non-trivial isotropy group is found at the origin. The principal stratum hence is the open set $\bC^2 \setminus \{0\}$, on which the action of $\psi_\tau$ is free, so that it becomes a bundle with fiber $\bS^1$. Nevertheless, the principal stratum $\bC^2 \setminus \{ 0\}$ does not have a global Poincar\'{e} section. This  would imply that it is a fiber bundle with base $\bS^1$, which is impossible since every loop in $\bC^2 \setminus \{ 0\}$ is contractible.
The action may be restricted to the invariant hyperboloids given by the level sets $\Im( z_1 \bar z_2) = L$. Whenever $L \not = 0$, the hyperboloid topologically is $\bR^2 \times \bS^1$ and there is a global section. Consequently upon removing the whole cone $\Im( z_1 \bar z_2) = 0$ from $\bC^2$, each connected component does have a global section. Alternatively one may remove one circle from each invariant $\bS^3$ defined by $|z_1|^2 + |z_2|^2 = const$, which turns $\bS^3$ into $\bR^2 \times\bS^1$. Which ``surgery'' to chose depends on properties of the map $f$.

When $l=1$ and $m>1$ there is a stratum in addition to the origin given by $(0, z_2)$ with discrete isotropy group $t = j / m$, $j = 1, \dots, m-1$. Every map with symmetry $\psi_\tau$ has this set as an invariant plane. Removing this plane gives the principal stratum $A \times \bC$, where $A = \bC \setminus \{ 0 \}$. The section $\Im(z_1) = 0$, $\Re(z_1) > 0$ has the set $z_1 = 0$ as an invariant boundary, and hence this section is a global Poincar\'{e} section for maps with symmetry $\psi$. Thus we find a trivial bundle $A \times \bC$ with fiber $\bS^1$ and base $\bR^+ \times \bC$.

When both $l$ and $m$ are bigger than 1 there are three strata: the origin $(0,0)$, the plane $(0, z_2)$ and the plan $(z_1, 0)$. Removing these from $\bC^2$ turns the manifold into $A^2$ where $A = \bC \setminus \{ 0 \}$. A section along the positive real axis in either punctured plane $A$ gives a global Poincar\'{e} section.

\section{Examples}\label{sec:Examples}
In this section we will give several examples to illustrate \Th{Reduction}.
We start with a classical two-dimensional example.

\bexam[(M\"obius Map)]
An elliptic M\"obius transformation is conjugate to the form
\[
	f(z) = \frac{\phantom{+}az+b}{-bz+a} \;,
\]
for real $a,b$ with $ab \neq 0$.
This map has fixed points at $z = \pm i$ and is analytic on the upper half plane $M = \{z\in \bC: \Im(z) > 0\}$.
As was noted by \cite{Maeda87}, $f$ has the symmetry $Y = (z^2+1)\frac\partial{\partial z}$, with flow
\[
	\psi_t(z) = \frac{\cos(t)z + \sin(t)}{\cos(t)-\sin(t)z} \;.
\]
On $M$ the orbits of $\psi$ are simply circles of radius $r$ and center $i\sqrt{r^2+1}$ for any $r \ge 0$. The action represented by the symmetry flow is not free, since $z=i$ is a fixed point of $\psi_t$. However on the principal stratum $M \setminus\{i\}$, the action is free, and has the global Poincar\'{e} section $\Sigma = \{i\sigma: 0<\sigma<1\}$, with return time $T = 2\pi$.
Notice that in order to achieve a global Poincar\'{e} section, the map the fixed point must be removed from the upper half plane.

To compute the reduced form \Eq{skewProduct}, we use the covering \Eq{covering}, $p(\sigma,\tau) = \psi_\tau(\sigma)$ and the notation \Eq{skewProduct}, $F(\sigma,0) = (k(\sigma),\rot(\sigma))$ on the section, to obtain
\begin{align*}
	f \circ p(\sigma,0) &= f(\sigma)=
	 \frac{a\sigma +b}{a-b\sigma}\\
	p \circ F(\sigma,0) &= \psi_{\rot(\sigma)}(k(\sigma)) =
	 \frac{\cos(\rot(\sigma))\,k(\sigma) + \sin(\rot(\sigma))}{\cos(\rot(\sigma))-\sin(\rot(\sigma))\,k(\sigma)} \;.
\end{align*}
Since, by \Eq{liftDef}, these must be equal we conclude that $F$ takes the form
\[
	F(\sigma,\tau) = (\sigma, \tau + \rot(\sigma)) \;,
\]
where $\rot(\sigma)\equiv\rot$ is, in fact, any constant such that $\rot = \arg( a + i b)$.
The most natural choice may be the principal value $\operatorname{Arg}( a + i b) \in (-\pi, \pi]$, and each branch of $\arg$ gives a different lift.
Note that both $f$ and $Y$ have the invariant
\[
	I(z) = \frac{1}{\Im(z)} (|z|^2+1) \;.
\]
According to \Lem{Invariant}, an invariant for $k$ is naturally $\iota^*I = \sigma + \frac{1}{\sigma}$, or trivially $\sigma$ itself.
\eexam

\bexam[(Twisted Symmetries):]
On the manifold $M=\bR^2\times (\bR/\bZ)$, define
\beq{moebius}
	f(\xi,z)= \left( R_{\beta z}\,h\left( R_{-\beta z} \xi \right),\, z
		+\alpha(|\xi|) \right) \;,
\eeq
where $\alpha: \bR^+ \to \bR$, $R_\theta \in SO(2)$ is the rotation by angle $2\pi \theta$, and $\beta \in \bR$. The map $h: \bR^2 \to \bR^2$ is smooth and, so that $f$ be continuous on $M$, is assumed to have the symmetry
\beq{twistedMap}
	h\circ R_\beta = R_\beta \circ h\;.
\eeq
Note that if $h$ were to commute with $R_\theta$ for all $\theta$, then $R_{\beta z} h(R_{-\beta z} \xi) = h(\xi)$, so \Eq{moebius} would already be written in the skew-product form. Instead we assume that $\beta$ is rational, so that the symmetry \Eq{twistedMap} is discrete.

The map \Eq{moebius} has the symmetry
\beq{rotationalFlow}
	\psi_t(\xi,z)=(R_{\beta t}\,\xi ,z+t) \;.
\eeq
This flow has a global Poincar\'{e} section $\Sigma = \{(\sigma,0): \sigma \in \bR^2\}$, with the covering map $p(\sigma,\tau) = \psi_\tau(\sigma,0)$.
Since $\Sigma$ is simply connected, \Th{Reduction} implies there exits a lift $F$ which can be easily computed using \Eq{kDefine}:
\beq{twistedLift}
	F(\sigma,\tau) = (R_{-\beta\alpha(|\sigma|) }\,h(\sigma), \tau + \alpha(|\sigma|)) \;.
\eeq
Since $\beta = \frac{p}{q}$ is rational, then by \Eq{twistedMap} $h$ must have the associated \emph{discrete} $(p,q)$ symmetry. In this case, the orbits of $\psi$ are $(p,q)$ torus knots on $M \setminus \{\xi = 0\}$. The action of the symmetry group is not free since the isotropy group of any point on the $z$-axis is different from those not on the axis. This example shows that even when the orbit space $M/G$ is not a manifold, the reduced map of \Th{Reduction} may still be as smooth as the original map. The point is that \Eq{twistedLift} has a residual discrete symmetry with the action $g(\xi,\tau) = (R_\beta\xi,\tau)$. Factoring out this discrete symmetry would indeed lead to a singular reduced space.

An example of an orbit of the reduced map in the section $\Sigma$ is shown in \Fig{sandDollar}, and many other examples of such maps can be constructed by replacing $h$ with any map with a discrete symmetry like those in \cite{Field09}.

\InsertFig{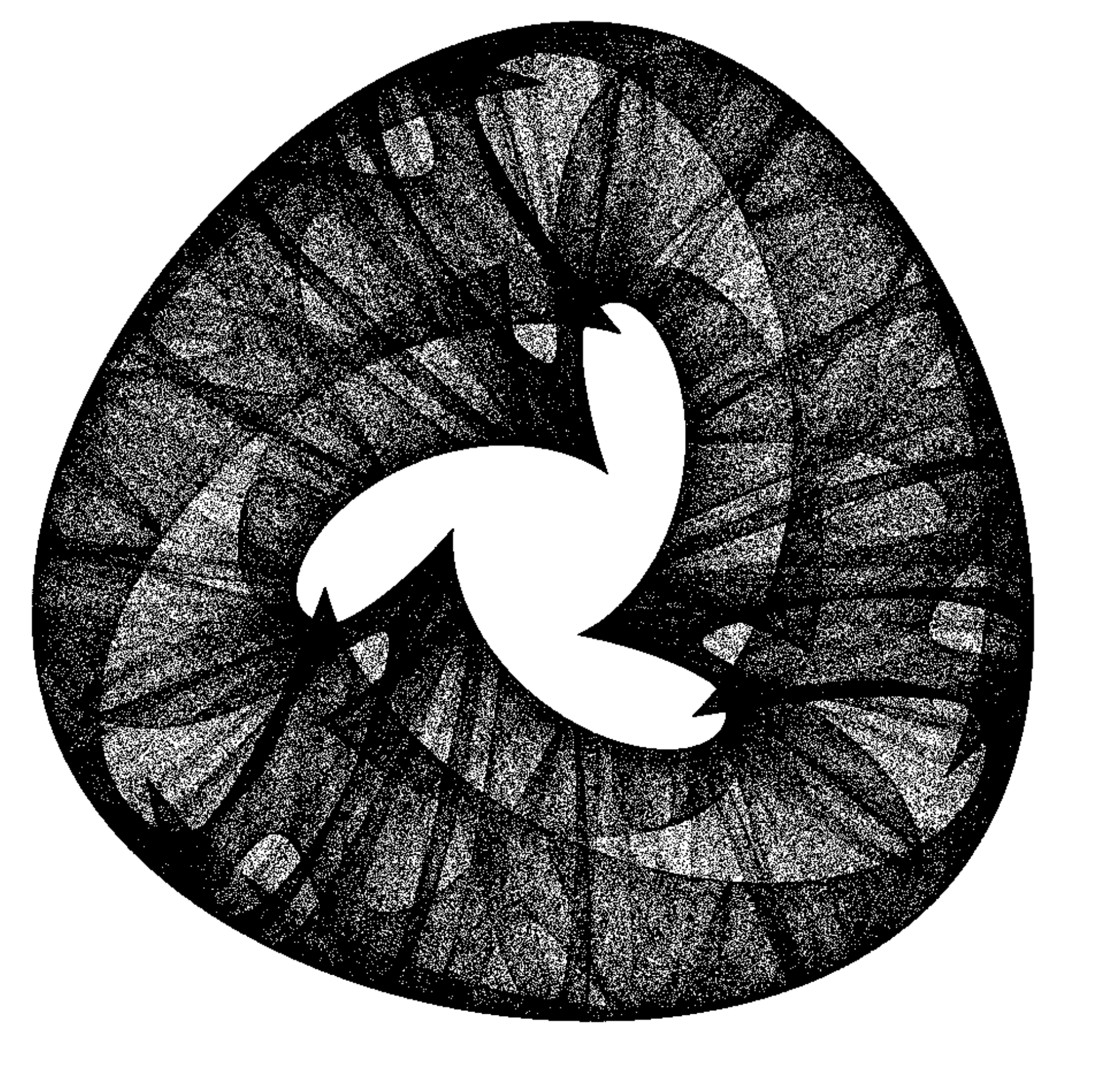}{Points on an attractor for the reduced map $k$ in \Eq{twistedLift} for $\beta = \frac{1}{3}$. Here, using the complex form $u = x+iy$, following \cite{Chossat88}, we took $h(u,\bar u) = (\lambda + u \bar u)u + \gamma \bar u ^2$.
The figure shows the case $\lambda = -2.43$, $\gamma = 0.1$ and $\alpha(|\sigma|) = - 2.67$.}{sandDollar}{3in}
\end{example}

\bexam[(Non-orientable manifold):]
For $(\xi, z) \in \bR^2 \times \bR$, let $G$ be the discrete group generated by $g(\xi,z)=(-\xi,z+1)$. Note that  $G \simeq \bZ$, and the action $(g,(\xi,z))\mapsto g(\xi,z)$ on $\bR^3$ is free and proper. This implies that $M=\mathbb{R}^3/G$
is a manifold; in this case $M$ is non-orientable.

Let $f:\mathbb{R}^3\to\mathbb{R}^3$ be the map
\beq{moebius2}
	f(\xi,z)= \left( R_{(z +\beta)/2}h\left( R_{-z/2}\xi \right),\, z   +\beta \right).
\eeq
where (as in the previous example) $R_\theta$ is the rotation by angle $2 \pi \theta$, $h: \bR^2 \to \bR^2$, and $\beta\in\mathbb{R}$. Note that since  $R_{(z+1)/2} = -R_{z/2}$, $f \circ g = g \circ f$, thus $f$ can be thought of as a map on the quotient $M$.

It is easy to see that the flow
\[
	\psi_t(\xi,z)=(R_{t/2} \xi,z+t) \;.
\]
is a symmetry of \Eq{moebius2}. Note that $\psi_{t}\circ g=g\circ\psi_{t}$ so that $\psi_t$ defines a flow on $M$. Moreover, since $\psi_{t+1}=g\circ\psi_t$ the orbits of the symmetry $\psi$ are embedded circles in $M$, and $\psi$ has the global Poincar\'{e} section
\[
	\Sigma=\{(\sigma,0): \sigma \in \bR^2\} \subset M\;.
\]
The return map for this case is simply $r_\psi(\sigma) = \psi_1(\sigma) = -\sigma$. Computation of the reduced map using \Eq{kDefine} gives  $F(\sigma,\tau) = (h(\sigma), \tau + \beta)$.

Note that \Eq{moebius2} is volume preserving on $M$ whenever $h$ is area preserving; for example, the map
\[
	h(x,y)= \left(y, \frac{2a y}{y^2+1}-x\right)
\]
is area preserving for any $a \in \bR$. This map also has the integral
\[
	J(x,y)=y^2 x^2+x^2-2a y x+y^2 \;,
\]
i.e., one has $J\circ h= J$.  In this case the function
\[
	I(\xi,z)=J\left(R_{-z/2} \xi\right) \;.
\]
is an integral of $f$. Since it is also an integral of $\psi_t$,  $I\circ \psi_t=I$, \Lem{Invariant} implies that $h$ has the reduced integral $\iota^* I = J$, as we already knew.
\eexam

\bexam[(Nonhyperbolic Cat Maps):]
Suppose $f$ is a diffeomorphism of the $n$-torus $M=\bR^n/\bZ^n$ that fixes the origin, $f(0) = 0$, and is homotopic to the map $b(\xi)=B\xi$, where $B\in SL(n,\bZ)$ is an $n\times n$ unimodular matrix.

We will assume that $B$ has a simple eigenvalue $\lambda = 1$ with right eigenvector $v$
and left eigenvector $w$. The vector field $V = v \cdot \nabla$ generates the flow
\beq{CatFlow}
	\psi_t(\xi)=\xi+t\,v \;,
\eeq
In order that $V$ be a symmetry of $f$, condition \Eq{MapSymmetry} requires that $Df(\xi)v=v$
for all $\xi\in M$. For example, the map
\beq{HomoCat}
	f(\xi)=B\xi+\phi(\xi) v \;,
\eeq
with $\phi(0) = 0$, is of the assumed form. It has the symmetry $\psi_t$ if $\phi(\xi+t\,v)=\phi(\xi)$, for all $t$, so that $D\phi(\xi) v = 0$.

Since $w$ is a left eigenvector of $B$, the set $\Sigma = \{\xi \in M : w \cdot \xi = 0 \mod 1\}$ is a codimension-one subspace invariant under $b$ and transverse to $v$. Indeed since $B$ has integer coefficients, $\Sigma\hookrightarrow M$ is a codimension-one torus; an example is sketched in \Fig{3torus}. Moreover, it is clear that since $\Sigma$ is transverse to $v$, it is a Poincar\'e section for the flow \Eq{CatFlow}.

Since $f$ and $b$ are homotopic and both fix the origin, $f_*=b_*$, as maps on the fundamental group $\fun(M,0)$.
Moreover, $\Sigma$ is an invariant set for $b$, so the group $\iota_*\fun\left(\Sigma,0\right)$ is invariant under
 $b_*$. Consequently,
\[
	f_*\iota_*\fun\left(\Sigma,0\right)=b_*\iota_*\fun\left(\Sigma,0\right)
				\subseteq  \iota_*\fun\left(\Sigma,0\right) \;.
\]

Therefore, \Th{Reduction} and condition \Eq{newSubset} imply that there exists a lift $F$
of $f$ to the cover $\Sigma\times\bR$ of the form of \Eq{skewProduct} with a covering map $p:\Sigma\times\bR\to M$ of the form $p(\sigma,\tau) 	 =  \sigma + \tau v$. To simplify the computations we will use an equivalent cover,
$\widetilde{p}:\bT^{n-1}\times \bR\to M$ given by
\[
	\widetilde{p}(\sigma,\tau) 	 = S \sigma + \tau v\;,
\]
where $(S\,|\,v) \in SL(n,\bZ)$ and the columns of the $n \times(n-1)$ integer matrix $S$ form a basis for $\Sigma$. Consequently,
\beq{Bhat}
	B (S\,|\,v) =  (S\,|\,v) \begin{pmatrix} \hat B & 0 \\ 0 & 1 \end{pmatrix}
\eeq
with $\hat B \in SL(n-1,\bZ)$.
In this case, a lift of the form \Eq{skewProduct} exists and must satisfy \Eq{liftDef}. For example, a lift of \Eq{HomoCat} is
\[
    F(\sigma,\tau) = (\hat B \sigma, \tau +\omega(\sigma) ) \;,
\]
where $\omega(\sigma)=\phi(S\sigma)$.

\InsertFig{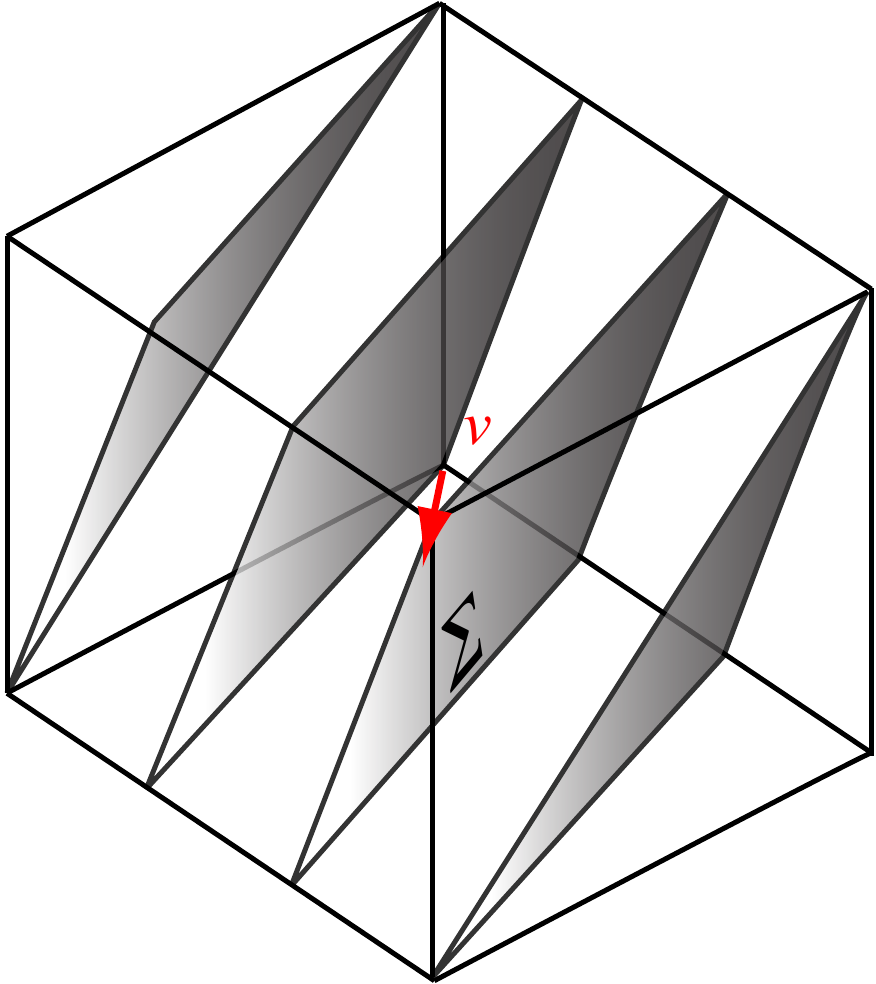}{The global Poincar\'{e}  section $\Sigma$ of the flow \Eq{CatFlow} for the matrix \Eq{CatMap} is an embedded submanifold of the three-torus $M$.}{3torus}{2.4in}

An explicit case of the form \Eq{HomoCat} for $n=3$ is
\beq{CatMap}
	B=\begin{pmatrix}	0 & 1 & 0 \\
 						0 & 0 & 1 \\
 						1 & -4 & 4
 		\end{pmatrix} \;,
		\quad 	\mbox{and } \phi(x,y,z) = g(x-y,y-z)
\eeq
for a function $g: \bT^2 \to \bR$ such that $g(0)=0$.
Note that $1$ is an eigenvalue of $B$ with right and left eigenvectors $v=(1,1,1)$ and $w = (-1,3,-1)$, respectively.
Thus the two-torus
\beq{catSection}
	\Sigma=\{(x,y,z)\in M: x-3y+z=0 \mod 1\} \;,
\eeq
shown in \Fig{3torus}, is invariant under $b$ and is a global Poincar\'{e} section for the flow \Eq{CatFlow}. Moreover, since $D\phi v = 0$, the map $f$ has symmetry \Eq{CatFlow}.

The fundamental group of $\Sigma$ is generated by the loops $\eta_1(t)=(2t,t,t)$ and $\eta_2(t)=(t,t,2t)$. Setting $u_i=\iota_*[\eta_i]$ for $i=1,2$, it is easy to verify that
$f_*u_1=u_2$ and $f_*u_2= -u_1+ 3u_2$. Consequently, \Th{Reduction} implies there exists a lift of $f$ of the form \Eq{skewProduct}. Indeed, the covering map
\[
	\widetilde{p}(\sigma,\tau) = \begin{pmatrix}
						2 & -1 & 1 \\ 1 & -1 & 1 \\ 1 & -2 & 1\end{pmatrix}
					\begin{pmatrix} \sigma_1 \\ \sigma_2 \\ \tau\end{pmatrix} \;,
\]
gives $(\sigma_1, \sigma_2, \tau) = (x-y,y-z,-x + 3y-z)$, so that $\tau = 0$ corresponds to $\iota(\Sigma)$. Upon computing $\hat B$ from \Eq{Bhat}, the lift takes the form \Eq{skewProduct}
with
\[
\begin{split}
	k(\sigma) &= \begin{pmatrix} 0 & 1 \\ -1 & 3 \end{pmatrix} \sigma \;,\\
	\rot(\sigma) &= g(\sigma_1,\sigma_2) \;.
\end{split}
\]
Since the eigenvalues of $k$ are $\gamma^{\pm 2}$ where $\gamma$ is the golden mean, the reduced map is an Anosov diffeomorphism.

Note that the one cannot freely replace \Eq{catSection} with any global section of \Eq{CatFlow} and still satisfy the topological requirement \Eq{lift}. For example, the torus $\tilde \Sigma = \{(0,y,z) \in M\}$ is also a global section for \Eq{CatFlow}. However, the map $b(\xi) = B\xi$ takes the generators $[(0,t,0)]$ and $[(0,0,t)]$ of the fundamental group of $\tilde \Sigma$ into $[(t,0,-4t)]$ and $[(0,t,4t)]$, violating \Eq{lift}.
\eexam

\bexam[(Resonant Circle Action):]
The flow $\psi_\tau(z_1,z_2) = (e^{il \tau} z_1, e^{i m \tau} z_2)$ of \Eq{circleAction}
corresponds to a circle action on $M = \bC^2$ and is familiar from the study of resonant, coupled oscillators.
As already noted in \Sec{OrbitSpaces} the action \Eq{circleAction} is not free: it has a fixed point at $(0,0)$ and points with nontrivial, discrete isotropy for each point $(z_1, 0)$ whenever $l>1$ and for each point $(0, z_2)$ whenever $m > 1$.

The general map on $\bC^2$ with the symmetry \Eq{circleAction} can be written
(see, e.g., \cite{Golubitsky88})
\[ \begin{split}
	z'_1 = f_1(\rho) z_1 + f_2(\rho)\bar z_1^{m-1} z_2^l \;, \\
	z'_2 = f_3(\rho) z_2 + f_4(\rho) z_1^m \bar z_2^{l-1} \;.
\end{split}\]
where the $f_j(\rho)$ are complex-valued functions of the real invariants of \Eq{circleAction}, namely
\beq{HilbertBasis}
   \begin{array}{ll}
      \rho_1 = z_1 \bar z_1 \;, & \rho_2 = z_2 \bar z_2 \;, \\
      \rho_3 = \Re( z_1^m \bar z_2^l) \;, & \rho_4 = \Im( z_1^m \bar z_2^l) \;.  \\
   \end{array}
\eeq
These four invariants form a \emph{Hilbert basis}: every function invariant under the action \Eq{circleAction} is a function of $\{\rho_1,\rho_2,\rho_3,\rho_4\}$ (see also \Sec{Comparison}). For example, $\det{Df}$ is invariant under the symmetry $\psi_\tau$, and hence a function of the invariants.

As a particular example consider
\[
 	(z'_1, z'_2) = ( z_1 - i \eps m \bar z_1^{m-1} z_2^l,
	              z_2 - i \eps l z_1^m \bar z_2^{l-1}) \;.
\]
This map is, to lowest order in $\eps$, the time-$\eps$ flow of the Hamiltonian $H =\rho_3$, and so is approximately a four-dimensional symplectic map in $(x_1,x_2,y_1,y_2)$, up to terms of order $\eps^2$. The flow of this Hamiltonian has two invariants, namely $H$ itself, and the Hamiltonian that generates the symmetry, $l \rho_1 + m \rho_2$. When $\eps \ll 1$ these functions will be approximate invariants of the map as well.

Assuming, for concreteness, that $m > l \ge 1$, then since $\psi_t$ does not act freely on $M=\bC^2$, there are points with nontrivial isotropy, in particular the set $\{(0, z_2): z_2 \in \bC\}$.  This set is forward invariant under both $\psi_t$ and $f$, and
removing it from $M$ gives the manifold $\tilde M = A \times \bC \simeq \bR^+ \times\bS^1 \times \bC$. On this manifold $\psi_t$ has the global Poincar\'e section
\[
	\Sigma = \{(x_1,z_2), x_1> 0, z_2 \in \bC\} \subset \tilde M \;.
\]
Thus to get a global section, it is not necessary to restrict to the principal stratum, which would entail, when $l>1$, the removal of the points $(z_1,0)$ as well. On the section $\Sigma$ these additional points are fixed under $f$.

The restriction of $f$ to $\tilde M$ has a lift on $\Sigma \times \bR$ determined by \Eq{kDefine}, i.e., by the requirement
\[
	e^{-i\rot} (x_1 -  i \eps m x_1^{m-1} z_2^l,
	z_2 - i \eps l x_1^m z_2^{l-1}) \in \Sigma \;.
\]
Thus $\rot$ is determined by requiring the first component to be real and positive, and $F$ is given by
\[\begin{split}
	k(x_1, z_2) &= ( x_1 |1 - i \eps m x_1^{m-2} z_2^l |,
	              e^{-i \rot} (z_2 - i \eps  l x_1^m z_2^{l-1} ) ) \;,\\
	\rot(x_1,z_2) &= \arg( 1- i \eps m x_1^{m-2} z_2^l) \;.
\end{split}\]

Note that even though the set $z_1 = 0$ is forward invariant, since $f$ is not invertible
there are points that map into the invariant set.
Similarly, the reduced map $k$ has points that map into the boundary $x_1 = 0$
given by the solutions of $z_2^l =  ( i \eps m x_1^{m-2})^{-1} $. The reduced map $k$ can be extended to the excluded set $x_1 =0$ by continuity, so that the pre-images
of $x_1 = 0$ have well defined orbits. Even so, the fibre map $\rot$ is undefined for points that map into the set $x_1 = 0$ so, strictly speaking, reduction by lifting fails in this case. However, even though the fibre map is undefined for certain points, the reduced map $k$ is well behaved.

%

\eexam

\section{Volume-Preserving Symmetry Reduction} \label{sec:VP}

We will now show that, when the map $f$ and its symmetry are both volume preserving, the reduced map $k$ of \Eq{skewProduct} is also volume preserving on $\Sigma$, with respect to an appropriate volume form.
This specializes \Th{Reduction} to the volume-preserving setting.
We denote the volume form by $\Omega \in \Forms^n(M)$; by assumption, both $f$ and $\psi_t$ preserve this form, $f^*\Omega = \Omega$ and $\psi_t^* \Omega = \Omega$, respectively. Equivalently, the symmetry vector field is incompressible: $L_Y \Omega = (\nabla \cdot Y) \Omega = 0$, where $\nabla \cdot Y$ is the divergence of $Y$. When $\nabla \cdot Y=0$,
we will say that $Y$ is incompressible.

\begin{thm}\label{thm:ReducedVolumeForm} In addition to the hypotheses of \Th{Reduction}, assume that $f$ is volume preserving and its symmetry $Y$ is incompressible. Then the reduced map $k$ of \Eq{skewProduct} preserves a volume form $\nu$ on $\Sigma$ defined by
\beq{nuDefn}
	\nu = \iota^*i_Y \Omega \;.
\eeq
\end{thm}
\begin{proof}
First we note that $\mu := i_Y \Omega$ is an $(n-1)$-form on $M$. Moreover, since $\Sigma$ is a section, $Y \not\in T\Sigma$, and so $\nu = \iota^*\mu$ is non-degenerate on $\Sigma$. In addition,
\(
	d\nu = \iota^* di_Y\Omega = \iota^*(L_Y\Omega-i_Yd\Omega) = 0 \;;
\)
thus $\nu$ is a volume form. Finally, from \Eq{MapSymmetry} we have
\(
	f^*\mu = i_{f^*Y} f^*\Omega = i_Y \Omega = \mu \;,
\)
so $f$ preserves $\mu$.

We now assert that, in fact,
\beq{form-aux}
	\Pi^*\nu = p^*\mu \;,
\eeq
where $\Pi:\Sigma\times \bR\to \Sigma$ is the canonical projection.
To see this recall that since $p$ is a local diffeomorphism and $Y\neq0$ everywhere,
there exists a flow box \cite{Abraham78} in $M$ near each point of $\Sigma$.
Therefore, we can reduce the proof to the case in which $\Sigma$ is an open set in
$\bR^{n-1}$, $M$ is of the form $M=\Sigma\times\bR$
and $p$ is the identity. If we let $(\sigma,\tau)$ be the coordinates of $M$
then, by construction, we have that $Y=\partial/\partial\tau$ and
the flow on $M$ is $\psi_t(\sigma,\tau)=(\sigma,\tau+t)$. In these coordinates,
we have that $\iota\circ\Pi(\sigma,\tau)=(\sigma,0)$, and the volume form
can be written as
\[
	\Omega=\kappa(\sigma,\tau)d\sigma_1\wedge\cdots\wedge d\sigma_{n-1}\wedge d\tau \;.
\]
Since $\psi_t$ is volume preserving, we conclude that
$\kappa(\sigma,\tau)=\kappa(\sigma,0)$, for all $(\sigma,\tau)\in M$.
This implies that $\mu=i_Y \Omega $ does not depend on $\tau$ and
therefore
\(
	p^*\mu =(\iota\circ\Pi)^*\mu=\Pi^*\nu \;,
\)
which is equality \Eq{form-aux}.

Finally from the definition \Eq{reducedMap} of $k$, and the fact that   $\Pi\circ j=id_\Sigma$,
we get that
\[
 	k^*\nu=j^* F^*\Pi^*\nu=j^* F^*p^*\mu
		 =j^*p^*f^*\mu=j^*p^*\mu=j^*\Pi^*\nu=\nu \;.
\]
In this way, we conclude that $k$ preserves $\nu$.
\end{proof}

Theorem \ref{thm:ReducedVolumeForm} can be combined with \Cor{FlowReduction} to show that $\nu$ is a reduced volume form for an incompressible vector field with an incompressible symmetry.

\begin{cor}\label{cor:ReducedVolumeForm}
Suppose that $X,Y \in \VectorField(M)$ are incompressible, commuting vector fields and that $Y$ has a flow $\psi$ with a global Poincar\'{e}  section $\Sigma$ that is an immersed manifold $\iota:\Sigma\hookrightarrow M$. Then there exists a vector field
$K$ on $\Sigma$ and a function $\zeta: \Sigma \to \bR$, such that
\beq{decomposition}
	\left.X\right|_\Sigma=K+ \zeta \left.Y\right|_\Sigma \;,
\eeq
where $K$ is incompressible with respect to the volume form $\nu=\iota^*i_Y\Omega$ on $\Sigma$.
\end{cor}

\begin{proof}
By \Th{Reduction} the flow $\varphi$ of $X$ has a lift $\Phi$ to the cover $\Sigma\times \bR$ of the form \Eq{FlowReduction}. Let $K \in \VectorField(\Sigma)$ denote the vector field generated by the reduced flow $k$ of $\Phi$. By
\Th{ReducedVolumeForm}, each $k_t$ preserves the volume-form $\nu$ so that
$K$ is incompressible on $\Sigma$ with respect to the form in \Eq{nuDefn}.

Finally, since $K$ is tangent to---and $Y$ is transverse to---$\Sigma$, the
vector field $X|_\Sigma$ can be written as a linear combination of $K$ and $Y$.
Indeed \Eq{FlowReduction} gives
\[
	\left.\frac{\partial}{\partial t}\right|_{t=0}\Phi_t(\sigma,\tau)=
		K + \left.\frac{\partial}{\partial t}\right|_{t=0} \omega(\sigma,t)\frac{\partial}{\partial \tau} \;.
\]
Since $Y=\frac{\partial}{\partial \tau}$ in the coordinates $(\sigma,\tau)$, this implies \Eq{decomposition} with $\zeta(\sigma)= \dot\omega(\sigma,0)$.
\end{proof}

A special case of this corollary is included in a theorem of Haller and Mezic \cite{Haller98}, who treated the case $M = \bR^3$ and noted that the flow on the two-dimensional group orbit space is Hamiltonian.

\bexam[(Volume-Preserving normal form):]

Consider the family of volume-preserving diffeomorphisms, $f: \bT^3 \to \bT^3$, defined by
\beq{VPNormal}
	f(x,y,z) = (x+a(y), y+z+b(y),z+b(y)) \;.
\eeq
Similar maps on $\bR^3$,
arise as the normal form for a volume-preserving map near saddle-node bifurcation with a triple-one multiplier \cite{Dullin08a}.  The map \Eq{VPNormal} has symmetry $Y = \frac\partial{\partial x}$ which generates the flow $\psi_t(x,y,z) = (x+t,y,z)$.
The natural section for $\psi$ is $\Sigma = \{(0,v,w): v,w\in \bT^2\}$, and the corresponding covering map $p: \bT^2 \times \bR \to \bT^3$ defined by $p(v,w,\tau) = \psi_\tau(0,v,w)$. The group $\iota_*\fun(\Sigma,0)$ is generated by the equivalence classes $\gamma_1 = [(0,t,0)]$ and $\gamma_2 = [(0,0,t)]$.
Since $a$ and $b$ are a periodic functions,
\[
f_*(\gamma_1) = \gamma_1 	\;, \quad f_*(\gamma_2) = \gamma_1 + \gamma_2 \;,
\]
If we assume that $a(0)=b(0)=0$, $f$ fixes the origin, and then the requirement \Eq{newSubset} holds.

Of course \Eq{VPNormal} is already in skew-product form, with $\tau = x$, and $\rot(v,w) = a(v)$, and the reduced map becomes $k(v,w) = (v+w+q(v), w+q(v))$. The latter is a generalized Chirikov standard map and is typically chaotic.

Notice that if $a(y) = y$ the topological condition would fail for any choice of section. One way to see this is to note that the linear part of the map would then be a Jordan block that does not have any $2$-dimensional invariant subspaces. Another is to note that lifting with the natural section would give $\rot(u,v) = u$, which is not a continuous function from $\bT^2 \to \bR$. Finally, one might think that this problem could be repaired by first lifting $f$ to the universal cover $\bR^3$, eliminating the topological requirement. However, in this case almost all orbits of the linear map are unbounded--- there is no return map to any surface.
\eexam

\bexam[(4D Volume-preserving map):]
Consider a map $f = T_2 \circ T_1$ of  $\bR^4 \simeq \bC^2$ with coordinates $(z_1, z_2, \bar z_1, \bar z_2)$
given by the composition of the two ``shears''
\begin{align*}
	T_1( z_1, z_2) = \left( z_1( a + i \sqrt{ c_1 - a^2 +
			h_1(z_2 \bar z_2)/(z_1 \bar z_1)}), z_2 \right) \\
	T_2( z_1, z_2) = \left( z_1, z_2( b + i \sqrt{ c_2 - b^2 +
			h_2(z_1 \bar z_1)/(z_2 \bar z_2)}) \right),
\end{align*}
where $a,b,c_1,c_2$ are real constants, and $h_1$ and $h_2$ are non-negative real functions.
Each of the shears has constant Jacobian $c_i$, so $\det Df = c_1 c_2$. Thus if
$c_1c_2 = 1$, $f$ is volume preserving.
If $c_1 = c_2  =1$ the map $f$ is the composition of two symplectic maps.

Each shear commutes with the two symmetries $\psi_t(z_i) = e^{it} z_i$. Taken together this implies that $f$ has a pair of commuting symmetries, i.e., is equivariant under the $\bT^2$ action $\psi_{(t_1, t_2)}( z_1, z_2) = ( e^{it_1} z_1 , e^{i t_2} z_2)$.
The reduction of \Th{ReducedVolumeForm} can be carried out recursively, thus reducing $f$ to an area-preserving map.

To construct a global section, we must remove the sets of nontrivial isotropy from $M$: removing the fixed sets $z_1 = 0$ and $z_2 = 0$ leaves $A^2 = (\bC \setminus \{0\})^2$. On this set there is a global Poincar\'e section for the two-parameter action,  $\Sigma = \{(x_1,x_2): x_i > 0\} \subset A^2$.
The result of the reduction is to simply introduce polar coordinates in both complex planes, and give a fiber map with $-\rot = ( \arg( z_1') , \arg( z_2') )$.

The reduced map $k$ would be written in terms of $(x_1, x_2) \in \Sigma$.
Instead we use the two invariant radii $\rho_1 = z_1 \bar z_1$ and $\rho_2 = z_2 \bar z_2$
as coordinates we find
\[
 T_1( \rho_1, \rho_2) =( c_1 \rho_1 + h_1(\rho_2)), \rho_2) = (\rho_1', \rho_2)
\]
and similarly for $T_2$, so that the reduced area-preserving map (assuming $c_1 c_2 = 1$) is
\[
 	T_2(T_1(\rho_1,\rho_2)) = (c_1\rho_1+ h_1(\rho_2), c_2\rho_2 + h_2(\rho_1'))\;,
\]
which  is smooth even for $\rho_j = 0$.

\eexam

\section{Noether symmetries}\label{sec:Invariant}

For Hamiltonian flows and symplectic maps the close relationship between the existence of invariants and symmetries is well-known. Noether's theorem, for example, implies that a Lagrangian system with a symmetry acting in configuration space has an invariant. This relationship is exploited in the Liouville-Arnold construction of action-angle variables for an integrable $n$-degree-of-freedom Hamiltonian where each invariant generates a symmetry vector field. In this case, the involutive property of the invariants implies that the corresponding symmetry vector fields commute.

As we remarked in the introduction, analogues of Noether's theorem do not exist for maps in general nor for volume-preserving maps in particular. For example the skew product
\[
	f(x,y,\tau) = (k(x,y),\tau + \rot(x,y))
\]
on $\Sigma \times \bS^1$ has the obvious symmetry $Y =\partial/\partial\tau$ and preserves the volume form $\Omega = dx \wedge dy \wedge d\tau$ whenever $k$ is area preserving. If $k$ has no invariant, then neither does $f$. This would occur, e.g., when $\Sigma = \bT^2$, and $k(x,y) = (2x+y,x+y)$, the cat map.

In addition, a three-dimensional map could have an invariant, restricting its orbits to two-dimensional level sets, and yet have no symmetry: the map can be chaotic on the level sets. Examples include the trace maps \cite{Roberts94}, such as the Fibonacci map
\[
	f(x,y,z) = (y,z,-x+2yz) \;.
\]
Trace maps are volume preserving; many similar examples have also been constructed \cite{Gomez02}.
Consequently, symmetries do not necessarily give rise to invariants without some additional structure.

One sufficient additional structure corresponds to symplectic maps with Hamiltonian symmetries.
A symmetry $Y$ that is a Hamiltonian vector field is known as a \emph{Noether symmetry}.
Here we give a slight generalization of a theorem of \cite[App. 1]{Bazzani88},
which required the symplectic map to be a twist map. We show that one can drop this requirement if the map has a mild recurrence property.

\begin{thm}[Symplectic Noether]\label{thm:SymplecticNoether}
If $Y$ is a Noether symmetry of a symplectic map $f$ and $f$ has any recurrent orbits, then $f$ has an invariant. Conversely if $I$ is an invariant of a symplectic map $f$, then its Hamiltonian vector field is a symmetry of $f$.
\end{thm}

\begin{proof}
For this proof let $\mu$ denote the symplectic form (e.g., $\mu = dq \wedge dp$), and let $I$ be the Hamiltonian of the vector field $Y$: $i_Y\mu = dI$. Consequently
\(
	f^* i_{Y} \mu = d ( f^*I)
\),
but by the symmetry and symplectic conditions
\(
	f^* i_{Y} \mu = i_{f^*Y} f^*\mu = i_{Y} \mu = dI
\).
Consequently $d(f^*I-I) = 0$, so that
$
	I(f(x)) = I(x) +c
$
for some constant $c$. Suppose $x^*$ is a recurrent point, then there exists a subsequence $t_i \to \infty$ such that $f^{t_i}(x^*) \to x^*$. However, since
\(
	I(f^t(x)) = I(f^{t-1}(x))+c = \ldots = I(x) + tc
\),
this implies that
\[
 \lim_{i\to \infty} [I(f^{t_i}(x^*) - I(x^*)] = \lim_{i \to \infty} t_i c = 0 \;,
\]
so that $c = 0$. Thus $I$ is an invariant for $f$.
Conversely, if the symplectic map $f$ has an invariant $I$, then it generates a Hamiltonian vector field $Y$, and
\[
	 i_{f^*Y} \mu = f^*(i_Y \mu) = f^*(dI) = d(f^*I) = dI = i_Y \mu \;.
\]
Thus $i_{f^*Y -Y} \mu =0$, but since $\mu$ is non-degenerate, this can only occur when $f^*Y = Y$.
\end{proof}

The condition that $f$ has a recurrent orbit is necessary. The map
\[
	f(x,y) = (x+c, y + g(x)) \;,
\]
on $\bR^2$ is symplectic with the canonical form $dx \wedge dy$ and has the translation symmetry $Y =\partial/\partial y$. This symmetry is generated by the Hamiltonian $I(x,y) = x$, which, however, is not an invariant for $f$ for any $c \neq 0$. Note that $f$ has no recurrent orbits.
By contrast, if we take the same $f$ but suppose that it is a map on $\bS^1 \times \bR$, then $f$ may have recurrent orbits, but the symmetry vector field $\partial/\partial y$ is locally, but not globally Hamiltonian, and so does not generate an invariant.

\bexam[(Lyness Map):]
For any $a > 0$, the Lyness map \cite{Lyness45},
\beq{Lyness}
	f(x,y) = \left(y, \frac{a+y}{x}\right)
\eeq
is a diffeomorphism on the positive quadrant $M = \bR^2_+$ that
preserves the symplectic form $\mu = \frac{1}{xy} dx \wedge dy$.
It has the symmetry \cite{Beukers98}
\[
	Y = {xy} ( -\partial_y I, \partial_x I)
	  =-xy\frac{\partial I}{\partial y}\frac{\partial }{\partial x}
	   +xy\frac{\partial I}{\partial x}\frac{\partial }{\partial y} \;,
\]
where
\[
	I = \frac{(1+x)(1+y)(a+x+y)}{xy} \;.
\]
The function $I$ is an invariant for $f$.
Moreover, $Y$ is the Hamiltonian vector field of $I$ with
respect to $\mu$ and so $I$ is an invariant for $Y$. This is
in agreement with \Th{SymplecticNoether}.

The level sets of the invariant $I$ for $I > I_{min} = I(x^*,x^*)$ are topologically circles that intersect the line
$\Sigma = \{(\sigma,\sigma): \sigma > x^*\}$ with $x^* = \frac12(1+\sqrt{1+4a})$ exactly once. The symmetry is not free since the fixed point $(x^*,x^*)$ has nontrivial isotropy. However on the principal stratum $\tilde M = M \setminus\{(x^*,x^*)\}$, $\Sigma$ is a global Poincar\'e section for the flow of $Y$.

Given that $\Sigma$ is simply connected, $f$ can be reduced with a lift.
The lifted map on $\Sigma \times \bR$
is $F(\sigma,\tau) = (\sigma, \tau + \rot(\sigma))$.
An integral expression for $\rot(\sigma)$ can be obtained from the results in \cite{Beukers98}.

\eexam

In the next section, we will apply these results to a four-dimensional symplectic map.

\section{Comparison of reduction methods} \label{sec:Comparison}
In this section, we will illustrate three different approaches to reduction, namely
\begin{itemize}\setlength{\itemsep}{0mm}
 \item invariants,
 \item polar coordinates, and
 \item reduction by lifting.
\end{itemize}
We will use, as an example, a four-dimensional symplectic map with a rotational symmetry.

A symplectic map on $M=\bR^4 \simeq \bC^2$ of standard type with a symmetry can be constructed using the Lagrangian generating function
$S(z,z') = \frac12 |z - z'|^2 - W(z \bar z)$ with $z = x + i y \in \bC$ and $W:\bR\to\bR$ a $C^2-$function.
This generating function is invariant under the rotation $( e^{i \tau} z, e^{i \tau} z')$.
Introducing the conjugate momenta $w = p_x + i p_y$, then the corresponding 4D map, generated
by $w = -\partial_{\bar z} S$ and $w' = \partial_{\bar z'} S$, is
\beq{4DSymplectic}
 	f(z,w) = ( z' , w') = \left( z + w', w - z V( z \bar z) \right)
\eeq
where $V(x)=\frac{d}{dx}W(x)$.
The map \Eq{4DSymplectic} has the symmetry $\psi_t(z, w) = ( e^{it} z, e^{it} w)$, e.g., the flow \Eq{circleAction} with $k=l=1$. Note that $\psi_t$ is symplectic and is the flow of the ``angular momentum" Hamiltonian, $L = \Im( \bar z w)$. The discrete Noether theorem, \Th{SymplecticNoether}, implies that $L$ is an invariant of $f$ as well.

\paragraph*{Reduction using invariants:}
Since the symmetry $\psi_t$ is a special case of \Eq{circleAction}, the basis \Eq{HilbertBasis} with $(z_1,z_2) \to (z,w)$ and $k=l=1$ generates its real polynomial invariants.
Note that these $\rho_i$ are not independent: they obey the relation
\beq{relation}
	\rho_3^2 + \rho_4^2 = \rho_1 \rho_2 \;,
\eeq
and the inequalities
\[
	\rho_1 \ge 0 \;,\quad \rho_2 \ge 0 \;.
\]
The map $(z, w) \to (\rho_1,\rho_2,\rho_3,\rho_4)$ is called the \emph{Hilbert map}.
Rewriting \Eq{4DSymplectic} in terms of the invariants gives
\bsplit{Hilbert}
	\rho_1' &= \rho_1( 1 - V)^2 + \rho_2 + 2 \rho_3(1 - V) \;, \\
	\rho_2' &= \rho_2 - 2 \rho_3 V + \rho_1 V^2 \;,\\
	\rho_3' &= \rho_3 - \rho_1 V + \rho_2' \;,
\end{split}\eeq
where $V= V(\rho_1)$
and $\rho_4 = -L$, the negative of the constant angular momentum. This is volume preserving with respect to the form $\Omega = d\rho_1 \wedge d\rho_2 \wedge d\rho_3$ on $\bR^3 [ \rho_1, \rho_2, \rho_3]$ and has an invariant $\rho_4^2 = \rho_1\rho_2-\rho_3^2$ inherited from the relation \Eq{relation}.

The Hilbert coordinates satisfy the Poisson structure given by
$\{\rho_1, \rho_2 \} =4 \rho_3$, $\{ \rho_1, \rho_3 \} = 2\rho_1$, $\{ \rho_2, \rho_3\} = - 2\rho_2$ with Casimir $\rho_1 \rho_2 - \rho_3^2$.
Moreover, the reduced mapping \Eq{Hilbert} is a Poisson map: it preserves this reduced Poisson structure.
The level set of the Casimir is a smooth submanifold provided $\rho_4 \not = 0$. On each such a level set (a smooth symplectic leaf) the map is area preserving. The reduction by the symmetry and by its generating invariant reduces the dimension of the map by two.

However, the invariant set $\rho_4 = 0$ is a half-cone. This case is a singular reduction, because the action $\psi_t$ is not free, since the origin $(z, w) = (0,0)$ is a fixed point for the whole group.

From this construction, the map in the symmetry direction, in
the form $\tau' = \tau + \rot(\rho_i)$, can be obtained as follows. The Hilbert map from full phase space $T^*\bR^2$ to the invariants is not invertible. It has the symmetry orbits as fibers. As pre-image in the fiber we can choose
\[
	(z, w) =h(\rho) =: \left(\sqrt{\rho_1}, \sqrt{\rho_2}
			e^{i \arg(\rho_3 - i\rho_4)}\right) \;.
\]

The angle increment $\rot$ is now defined by $f(h(\rho)) = \psi_{\rot} h( \rho')$.
The explicit form of the fiber map therefore is found from $z' = e^{i \rot} \rho_1'$
so that $\rot = \arg( z')$. The map $h$ is not defined when $\rho_3 = \rho_4 = 0$, i.e.\ at the origin of the $z$-plane and at the origin of the $w$-plane.

\paragraph*{Reduction using polar coordinates:}
Polar coordinates appear whenever the symmetry angle $\theta$ becomes a coordinate. There is a unique extension of this transformation of coordinates to a symplectic transformation of all of phase space that is linear in the momenta, the so-called cotangent lift, see, e.g.~\cite{Abraham78}. Denoting the new variables by $(r, \theta, p_r, p_\theta)$, the momentum $p_\theta$, conjugate to $\theta$, becomes the conserved quantity.
Specifically set $z = r e^{i\theta}$ with cotangent lift $w = e^{i\theta} ( p_r - i p_\theta/r)$ and rewrite the map in the new coordinates $(r, \theta, p_r) = ( \sqrt{ z \bar z}, \arg z, \Re( z \bar w )/r)$.

Having done most of the work in for the Hilbert basis already, we proceed as follows. In terms of the invariants, the polar coordinates are $(r, p_r, p_\theta) = ( \sqrt{\rho_1}, \rho_3/\sqrt{\rho_1}, \rho_4 )$ and the map is obtained from that of the invariants by eliminating $\rho_2$ using $\rho_2 = (\rho_3^2 + \rho_4^2)/\rho_1 = p_r^2 + p_\theta^2/r^2$. This gives an area-preserving map that depends on the parameter $p_\theta$
\bsplit{PolarCoord}
	r'   &= \left( p_\theta^2 / r^2 + (r + p_r - r V)^2 \right)^{1/2} \;, \\
	p_r' &= \left( r'^2 - r( r + p_r - r V)\right) / r' \;,
\end{split}\eeq
where $V = V(r^2)$. The fiber map is obtained from rewriting $z' = z + w' = z + w - z V$ in polar coordinates,
hence $r' e^{i\theta'} = e^{i\theta}( r + p_r - i p_\theta/r - r V)$ and therefore
\[
 \theta' = \theta + \arg( r + p_r - r V - i p_\theta/r) \;.
\]
This map works well unless $r' = 0$, which makes $p'_r$ and $\theta'$ undefined.
Points with $r'=0$ can only be reached when $p_\theta = 0$, but in that case it happens
as soon as $r + p_r = r V$. Excluding $p_\theta = 0$ gives a well-defined, smooth map.

\paragraph*{Reduction by lifting:}
For reduction by lifting we need to find a global Poincar\'{e}  section of the symmetry
flow $\psi_t$. As a first attempt try the surface defined by $\Im(z) = 0$ for $x =\Re(z) > 0$. The flow is tangent to the section when $\Re(z) = 0$, hence all of the plane $\{z = 0\}$ is tangent to the section. Unfortunately, this plane of tangency is not
invariant under the map. In order to avoid this difficulty, we can choose a larger
invariant set that contains the tangency set, and restrict the map to the complement of
this set.

The angular momentum $L(z,w)=-\Im( z \bar w)=\Im( \bar z w)$ is invariant both under the flow and the symplectic map. Notice that the plane $\{z = 0\}$
is contained in the level set $\{L(z,w)=0\}$, that is invariant.
The manifold $M \setminus \{ L(z,w) = 0\}$ has two connected components,
\[
	\tilde M=\{(z,w)\in\bC\times\bC: L(z,w)>0 \} \;,
\]
and the corresponding negative angular momentum set. As the analysis for either component is the same, we will now restrict the map $f$ and the flow $\psi_t$ to $\tilde M$.
The symmetry flow as a global Poincar\'{e} section on $\tilde M$:
\[
	\Sigma=\{(x,w)\in \tilde M: x>0\} \;.
\]
This section is relatively closed in $\tilde M$ and if $(x,w)\in\Sigma$ then $\Im(w)>0$.

We know that the lift of $f$ is of the form
$F(\sigma,\tau)=(k(\sigma),\tau+\rot(\sigma))$. Taking a point $\sigma=(x, w) \in \Sigma$
we can compute its image under \Eq{4DSymplectic} to obtain
\beq{LiftReduced}
	f(x,w)=(x + w - x V(x^2), w - x V(x^2)) \;.
\eeq
By \Eq{kDefine}, this equals $\psi_{\rot(\sigma)}(k(\sigma))$, so that $e^{-i \rot} ( x + w - x V)$ must be real and positive (where $x$ is assumed positive). Noticing that $\Im( x + w - x V)>0$, this immediately gives $\rot(\sigma) = \arg( x + w - x V)$, and we can choose the branch so that $0<\rot(\sigma)<\pi$. The reduced map thus becomes
\[
	k(x,w) = ( |x + w - x V| , e^{-i \rot(x,w)} (w - x V)) \;.
\]
By \Eq{nuDefn}, $k$ preserves the volume form $\nu = p_y dx \wedge dp_x \wedge dp_y$.

Because we started with a symplectic map the reduced map has an invariant, $\iota^* \Im(\bar z w) =  x \Im( w) = x p_y = p_\theta$. If we were to eliminate $p_y$ by fixing the invariant, we would again obtain the polar coordinate map \Eq{PolarCoord}.

\paragraph*{Comparison:}
The Hilbert map is similar to what we call the projection $\Pi$ to the section $\Sigma$,
while its ``inverse'' $h$ is like the inclusion map $\iota$. In differential geometry, a fiber bundle is said to have a ``global section'' if $h$ is continuous.
Notice, however, that the Hilbert map has fiber $\bS^1$, while when we construct the
global Poincar\'{e} section we have a bundle with base $\bS^1$ and fiber $\Sigma$.
In fact, as pointed out in \Sec{CircleAction}, some finite covering of the bundle is a direct product with $\bS^1$. In the present example the manifold $\tilde M$ is
topologically a direct product $\bS^1 \times \bR \times \bR^+$.

The Hilbert map  may be applied to points with isotropy where projection to $\Sigma$ is not defined. The reduced map \Eq{Hilbert} is well-defined, even when the $\bS^1$-bundle does not possess a global section or when there is no global Poincar\'{e} section.
Reduction by lifting is more efficient in that one does not need to introduce extra coordinates; moreover, one does not need to know the ``good'' coordinate system to do the reduction.
Instead the choice of global Poincar\'{e} section defines such a coordinate system.


%

 \section{Conclusions}\label{sec:Conclusions}
We have studied maps with continuous symmetries, specializing to the case that the symmetries have a global Poincar\'e section. We showed in \Th{Reduction} that if a necessary topological condition is satisfied, then the map has a lift such that in certain coordinates the lift has the skew-product form \Eq{skewProduct}. We called this ``reduction by lifting" because the map on the section, the reduced map $k$, describes the dynamics modulo the symmetry. The fiber map, which corresponds to a translation, is obtained naturally from the symmetry flow.

The topological conditions for the existence of a lift require that the homotopy homomorphism induced by the map leaves the fundamental group of the Poincar\'e section invariant. This requirement is trivial if the section is simply connected. In principle, the restriction could be avoided if one first lifts the dynamics to the universal cover; however---as we saw in the examples---almost all orbits may become unbounded. In many cases, a minimal lift using only the symmetry flow seems more useful.

If the map and its symmetry are volume-preserving, then the reduced map is also volume-preserving as shown in \Th{ReducedVolumeForm}. If both the map and its symmetry preserve an invariant, then the reduced map also has an invariant. Finally, if the map and its symmetry are symplectic, then, as in Noether's theorem, the existence of a symmetry implies an invariant, providing the map has recurrent orbits, recall \Th{SymplecticNoether}.

A number of examples were given to compare and contrast reduction by lifting to standard reduction techniques. Reduction by lifting is more parsimonious than the Hilbert mapping technique, which can result in a high-dimensional map that must satisfy a number of constraints. It is more explicit than the standard reduction by group orbit technique, and moreover applies when the action of the symmetry is neither proper nor free. When the action is not free, reduction by lifting need not lead to a map on a singular space.

However, the existence of a global section for the symmetry flow is a strong restriction. This can be overcome in some cases, as we showed, by restriction of the dynamics to invariant strata of the symmetry flow.

\appendix

\section{Forms and Lie Derivatives}\label{app:notation}
Here we set out our notation, which follows, e.g., \cite{Abraham88}. The set of $k$-forms on a manifold $M$ is denoted by $\Forms^k(M)$, and the set of $C^1$ vector fields by $\VectorField(M)$. If $\alpha \in \Forms^k(M)$ and $V_1, V_2, \ldots V_k \in \VectorField(M)$, then the pullback, $f^*$, of a form $\alpha$ by a diffeomorphism $f$ is defined by
\beq{PullBackF}
	(f^*\alpha)_x(V_1,V_2,...,V_k) =\alpha_{f(x)}(Df(x)V_1(x),\ldots, Df(x)V_k(x)) \;.
\eeq
The pullback can be applied to a vector field $V$ as well:
\beq{PullBackV}
	(f^{*}V)(x) := (Df(x))^{-1}V(f(x)) \;.
\eeq
The push-forward operator is defined as
\(
	f_* = (f^{-1})^* \;.
\)
The inner product of $\alpha$ with $V$ is defined
as the $(k-1)$-form
\[
	i_V \alpha := \alpha(V,\cdot,\ldots,\cdot) \;.
\]
Suppose that $\vphi_t$ is the ($C^1$) flow of a vector field $V$, so that $\vphi_0(x) = x$,
and $\frac{d}{dt} \vphi_t(x) = V(\vphi_t(x))$.
Then the Lie derivative
with respect to $V$ is the linear operator defined by
\beq{LieDeriv}
	L_V (\cdot) := \left. \frac{d}{dt}\right|_{t=0} \vphi_t^* (\cdot)
\eeq
where $\cdot$ is any tensor. In particular for a vector field $X$,
\beq{LieBracket}
	L_V X = [V,X]
\eeq
is the Lie bracket. The Lie derivative acting on
differential forms obeys Cartan's homotopy formula
\beq{Cartan}
	L_V \alpha = i_V ( d \alpha) + d(i_V \alpha) \;.
\eeq


\bibliographystyle{alpha}
\bibliography{integrable}

\end{document}